\documentclass[10pt]{article}
\usepackage[utf8]{inputenc}
\usepackage{amsmath}
\usepackage{amsthm}
\usepackage{amssymb}
\usepackage{fullpage}
\usepackage{cite}
\usepackage{hyperref}
\usepackage{subcaption}
\usepackage{graphicx}
\usepackage[linesnumbered, nofillcomment,noend]{algorithm2e}
\usepackage{stackrel}

\newtheorem{theorem}{Theorem}
\newtheorem*{thm*}{Theorem}

\newtheorem{lemma}{Lemma}
\newtheorem{corollary}{Corollary}
\newtheorem{observation}{Observation}

\begin{document}

\title{\bf New bounds on the tile complexity of thin rectangles at temperature-1}

\author{%
David Furcy\thanks{Computer Science Department, University of Wisconsin Oshkosh, Oshkosh, WI 54901, USA,\protect\url{furcyd@uwosh.edu}.}
\and
Scott M. Summers\thanks{Computer Science Department, University of Wisconsin Oshkosh, Oshkosh, WI 54901, USA,\protect\url{summerss@uwosh.edu}.}
\and
Christian Wendlandt\thanks{Computer Science Department, University of Wisconsin Oshkosh, Oshkosh, WI 54901, USA,\protect\url{wendlc69@uwosh.edu}.}
}

\date{}
\maketitle

\begin{abstract}
In this paper, we study the minimum number of unique tile types required for the self-assembly of thin rectangles in Winfree's abstract Tile Assembly Model (aTAM), restricted to temperature-1. Using Catalan numbers, planar self-assembly and a restricted version of the Window Movie Lemma, we derive a new lower bound on the tile complexity of thin rectangles at temperature-1 in 2D. Then, we give the first known upper bound on the tile complexity of ``just-barely'' 3D thin rectangles at temperature-1, where tiles are allowed to be placed at most one step into the third dimension. Our construction, which produces a unique terminal assembly, implements a just-barely 3D, zig-zag counter, whose base depends on the dimensions of the target rectangle, and whose digits are encoded geometrically, vertically-oriented and in binary. 
\end{abstract}

\section{Introduction}
\label{sec:intro}
Intuitively, self-assembly is the process through which
simple, unorganized components spontaneously combine, according to
local interaction rules, to form some kind of organized final
structure.

While nature exhibits numerous examples of self-assembly, researchers have been investigating the extent to which the power of nano-scale self-assembly can be harnessed for the systematic nano-fabrication of atomically-precise computational, biomedical and mechanical devices. For example, in the early 1980s, Ned Seeman \cite{Seem82} exhibited an experimental technique for controlling nano-scale self-assembly known as ``DNA tile self-assembly''.

Erik Winfree's abstract Tile Assembly Model (aTAM) is a simple, discrete mathematical model of DNA tile self-assembly. In the aTAM, a DNA tile is represented as an un-rotatable unit square \emph{tile}. Each side of a tile may have a \emph{glue} that consists of an integer strength, usually 0, 1 or 2, and an alpha-numeric label. The idea is that, if two tiles abut with matching kinds of glues, then they bind with the strength of the glue. In the aTAM, a tile set may consist of a finite number of tiles, because individual DNA tiles are expensive to manufacture. However, an infinite number of copies of each tile are assumed to be available during the self-assembly process in the aTAM. Self-assembly starts by designating a \emph{seed} tile and placing it at the origin. Then, a tile can come in and bind to the seed-containing \emph{assembly} if it binds with total strength at least a certain experimenter-chosen, integer value called the \emph{temperature} that is usually 1 or 2. Self-assembly proceeds as tiles come in and bind one-at-a-time in an asynchronous and non-deterministic fashion. 

Tile sets are designed to work at a certain temperature value. For instance, a tile set that self-assembles correctly at temperature-2 will probably not self-assemble correctly at temperature-1. However, if a tile set works correctly at temperature-1, then it can be easily modified to work correctly at temperature-2 (or higher). In what follows, we will refer to a ``temperature-2'' tile set as a tile set that self-assembles correctly only if the temperature is 2 and a ``temperature-1'' tile set as a tile set that self-assembles correctly if the temperature is 1.  

Temperature-2 tile sets give the tile set designer more control over the order in which tiles bind to the seed-containing assembly. 
For example, in a temperature-2 tile set, unlike in a temperature-1 tile set, the placement of a tile at a certain location can be prevented until after the placements of at least two other tiles at two respective adjacent locations. 
This is known as \emph{cooperative binding}.
Cooperative binding in temperature-2 tile sets leads to the self-assembly of computationally and geometrically interesting structures, in the sense of Turing universality \cite{Winf98}, the efficient self-assembly of $N \times N$ squares \cite{AdlemanCGH01,RotWin00} and algorithmically-specified shapes \cite{SolWin07}. 

While it is not known whether the results cited in the previous paragraph hold for temperature-1 tile sets, the general problem of characterizing the power of non-cooperative tile self-assembly is important from both a theoretical and practical standpoint. This is because when cooperative self-assembly is implemented in the laboratory \cite{RoPaWi04,BarSchRotWin09,SchWin07,MaoLabReiSee00,WinLiuWenSee98}, erroneous non-cooperative binding events may occur, leading to the production of invalid final structures. Of course, the obvious way to minimize such erroneous non-cooperative binding events is for experimenters to always implement systems that work in non-cooperative self-assembly because temperature-1 tile sets will work at temperature-1 or temperature-2. Yet, how capable is non-cooperative self-assembly, in general, or even in certain cases? At the time of this writing, no such general characterization of the power of non-cooperative self-assembly exists, but there are numerous results that show the apparent weakness of specific classes of temperature-1 tile self-assembly \cite{jLSAT1,ManuchSS10,WindowMovieLemma,MeunierW17}. 

Although these results highlight the weakness of certain types of temperature-1 tile self-assembly, if 3D (unit cube) tiles are allowed to be placed in just-barely-three-dimensional Cartesian space (where tiles may be placed in just the $z=0$ and $z=1$ planes), then temperature-1 self-assembly is nearly as powerful as its two-dimensional cooperative counterpart. For example, like 2D temperature-2 tile self-assembly, just-barely 3D temperature-1 tile self-assembly is capable of simulating Turing machines \cite{CookFuSch11} and the efficient self-assembly of squares \cite{CookFuSch11,jFurcyMickaSummers} and algorithmically-specified shapes \cite{FurcyS18}. 

Furthermore, Aggarwal, Cheng, Goldwasser, Kao, Moisset de Espan\'{e}s and Schweller \cite{AGKS05g} studied the efficient self-assembly of $k \times N$ rectangles, where $k < \frac{\log N}{\log \log N - \log \log \log N}$ at temperature-2 in 2D (the upper bound on $k$ makes such a rectangle \emph{thin}). They proved that the size of the smallest set of tiles that uniquely self-assemble into (i.e., the \emph{tile complexity} of) a thin $k \times N$ rectangle is $O\left(N^{\frac{1}{k}}+k\right)$ and $\Omega\left(\frac{N^{\frac{1}{k}}}{k}\right)$ at temperature-2. Their lower bound actually applies to all tile sets (temperature-1, temperature-2, etc.) but their upper bound construction requires temperature-2 and does not work correctly at temperature-1. 

In this paper, we continue the line of research into the tile complexity of thin rectangles, initiated by Aggarwal, Cheng, Goldwasser, Kao, Moisset de Espan\'{e}s and Schweller, but exclusively for temperature-1 tile self-assembly.

\subsection{Main results of this paper}

The main results of this paper are bounds on the tile complexity of thin rectangles at temperature-1. We give an improved lower bound for the tile complexity of 2D thin rectangles as well as a non-trivial upper bound for the tile complexity of just-barely 3D thin rectangles. Intuitively, a just-barely 3D thin rectangle is like having at most two 2D thin rectangles stacked up one on top of the other. We prove two main results: one negative (lower bound) and one positive (upper bound). Our main negative result gives a new and improved asymptotic lower bound on the tile complexity of a 2D thin rectangle at temperature-1, without assuming unique production of the terminal assembly (unique self-assembly). 

\begin{theorem}
\label{thm:one}
The tile complexity of a $k \times N$ rectangle for temperature-1 tile sets is $\Omega\left( N^{\frac{1}{k}}\right)$.
\end{theorem}

Currently, the best upper bound for the tile complexity of a $k \times N$ rectangle for temperature-1 tile sets is $N + k - 1$, and is obtained via a straightforward generalization of the ``Comb construction'' by Rothemund and Winfree (see Figure 2b of \cite{RotWin00}). So, while Theorem~\ref{thm:one} currently does not give a tight bound, its proof technique showcases a novel application of Catalan numbers to proving lower bounds for temperature-1 self-assembly in 2D, and could be of independent interest. 

Our main positive result is the first non-trivial upper bound on the tile complexity of just-barely 3D thin rectangles. 

\begin{theorem}
\label{thm:two}
The tile complexity of a just-barely 3D $k \times N$ thin rectangle for temperature-1 tile sets is $O\left(N^{\frac{1}{\left \lfloor \frac{k}{3} \right \rfloor}} + \log N \right)$. Moreover, our construction produces a tile set that self-assembles into a unique final assembly.
\end{theorem}

We say that our main positive result is the first non-trivial upper bound because a straightforward generalization of the aforementioned Comb construction would give an upper bound of $O\left(N + k\right)$ on the tile complexity of a just-barely 3D $k \times N$ thin rectangle.

\subsection{Comparison with related work}

Aggarwal, Cheng, Goldwasser, Kao, Moisset de Espan\'{e}s and Schweller \cite{AGKS05g} give a general lower bound of $\Omega\left( \frac{N^{\frac{1}{k}}}{k}\right)$ for the tile complexity of a 2D $k \times N$ rectangle for temperature-$\tau$ tile sets. Our main negative result, Theorem~\ref{thm:one}, is an asymptotic improvement of this result for the special case of temperature-1 self-assembly. 

Aggarwal, Cheng, Goldwasser, Kao, Moisset de Espan\'{e}s and Schweller \cite{AGKS05g} also prove that the tile complexity of a 2D $k \times N$ thin rectangle for general positive temperature tile sets is $O\left( N^{\frac{1}{k}}+k\right)$. Our main positive result, Theorem~\ref{thm:two}, is inspired by but requires a substantially different proof technique from theirs. Our construction, like theirs, uses a just-barely 3D counter, the base of which depends on the dimensions of the target rectangle, but unlike theirs, ours self-assembles in a zig-zag manner and the digits of the counter are encoded geometrically, vertically-oriented and in binary.

\section{Preliminaries}
\label{sec:prelims}

\newcommand{\lab}{{\rm label}}
\newcommand{\colorf}[1]{\color(#1)}
\newcommand{\strength}{{\rm str}}
\newcommand{\strengthf}[1]{\strength(#1)}
\newcommand{\bbval}[1]{\left[\!\left[ #1 \right] \! \right]}
\newcommand{\dom}{{\rm dom} \;}
\newcommand{\asmb}{\mathcal{A}}
\newcommand{\asmbt}[2]{\asmb^{#1}_{#2}}
\newcommand{\asmbtt}{\mathcal{A}^\tau_T}
\newcommand{\ste}[2]{#1 \mapsto #2}
\newcommand{\frontier}[3]{{\partial}^{#1}_{#2}{#3}}
\newcommand{\frontiert}[1]{\partial^{\tau}{#1}}
\newcommand{\frontiertt}[1]{\frontier{\tau}{t}{#1}}
\newcommand{\frontiertx}[2]{{\partial}^{\tau}_{#1}{#2}}
\newcommand{\frontiertau}[1]{{\partial}^{\tau}{#1}}
\newcommand{\arrowstett}[2]{#1 \xrightarrow[\tau,T]{1} #2}
\newcommand{\arrowste}[2]{#1 \stackrel{1}{\To} #2}
\newcommand{\arrowtett}[2]{#1 \xrightarrow[\tau,T]{} #2}
\newcommand{\res}[1]{\textrm{res}(#1)}
\newcommand{\termasm}[1]{\mathcal{A}_{\Box}[\mathcal{#1}]}
\newcommand{\prodasm}[1]{\mathcal{A}[\mathcal{#1}]}
\newcommand{\fgg}[1]{fgg^\#_{#1}}
\newcommand{\ftdepth}[2]{\textrm{ft-depth}_{#1}\left(#2\right)}
\newcommand{\str}[1][*]{\textrm{str}_{#1}}
\newcommand{\col}[1]{\textrm{col}_{#1}}
\newcommand{\tmblank}{\llcorner \negthinspace\lrcorner}
\newcommand{\fullgridgraph}{G^\mathrm{f}}\newcommand{\bindinggraph}{G^\mathrm{b}}

\newcommand{\setr}[2]{\left\{\ #1 \ \left|\ #2 \right. \ \right\}}
\newcommand{\setl}[2]{\left\{\ \left. #1 \ \right|\ #2 \ \right\}}

In this section, we briefly sketch a 3D version of Winfree's abstract Tile Assembly Model. Going forward, all logarithms in this paper are base-$2$.

Fix an alphabet $\Sigma$.
$\Sigma^*$ is the set of finite strings over $\Sigma$. Let $\mathbb{Z}$, $\mathbb{Z}^+$, and $\mathbb{N}$ denote the set of integers, positive integers, and non-negative integers, respectively. Let $d \in \{2, 3\}$. 

A \emph{grid graph} is an undirected graph $G=(V,E)$, where $V \subset \mathbb{Z}^d$, such that, for all $\left\{\vec{a},\vec{b}\right\} \in E$, $\vec{a} - \vec{b}$ is a $d$-dimensional unit vector. The \emph{full grid graph} of $V$ is the undirected graph $\fullgridgraph_V=(V,E)$,
such that, for all $\vec{x}, \vec{y}\in V$, $\left\{\vec{x},\vec{y}\right\} \in E \iff \| \vec{x} - \vec{y}\| = 1$, i.e., if and only if $\vec{x}$ and $\vec{y}$ are adjacent in the $d$-dimensional integer Cartesian space.

A $d$-dimensional \emph{tile type} is a tuple $t \in (\Sigma^* \times \mathbb{N})^{2d}$, e.g., a unit square (cube), with four (six) sides, listed in some standardized order, and each side having a \emph{glue} $g \in \Sigma^* \times \mathbb{N}$ consisting of a finite string \emph{label} and a non-negative integer \emph{strength}. We call a $d$-dimensional tile type merely a {\em tile type} when $d$ is clear from the context.

We assume a finite set of tile types, but an infinite number of copies of each tile type, each copy referred to as a \emph{tile}. A \emph{tile set} is a set of  tile types and is usually denoted as $T$.

A {\em configuration} is a (possibly empty) arrangement of tiles on
the integer lattice $\mathbb{Z}^d$, i.e., a partial function $\alpha:\mathbb{Z}^d
\dashrightarrow T$.  Two adjacent tiles in a configuration \emph{bind},
\emph{interact}, or are \emph{attached}, if the glues on their
abutting sides are equal (in both label and strength) and have
positive strength.  Each configuration $\alpha$ induces a
\emph{binding graph} $\bindinggraph_\alpha$, a grid graph whose
vertices are positions occupied by tiles, according to $\alpha$, with
an edge between two vertices if the tiles at those vertices
bind. For two
non-overlapping configurations $\alpha$ and $\beta$, $\alpha \cup \beta$
is defined as the unique configuration $\gamma$ satisfying, for all
$\vec{x} \in \dom{\alpha}$, $\gamma(\vec{x}) = \alpha(\vec{x})$, for
all $\vec{x} \in \dom{\beta}$, $\gamma(\vec{x}) = \beta(\vec{x})$, and
$\gamma(\vec{x})$ is undefined at any point $\vec{x} \in \mathbb{Z}^d
\backslash \left( \dom{\alpha} \cup \dom{\beta} \right)$.

An \emph{assembly} is a connected, non-empty configuration,
i.e., a partial function $\alpha:\mathbb{Z}^d \dashrightarrow T$ such that
$\bindinggraph_{\dom \alpha}$ is connected and $\dom \alpha \neq
\emptyset$. Given $\tau\in\mathbb{Z}^+$, $\alpha$ is \emph{$\tau$-stable} if every cut-set
of~$\bindinggraph_\alpha$ has weight at least $\tau$, where the weight
of an edge is the strength of the glue it represents.\footnote{A
  \emph{cut-set} is a subset of edges in a graph which, when removed
  from the graph, produces two or more disconnected subgraphs. The
  \emph{weight} of a cut-set is the sum of the weights of all of the
  edges in the cut-set.} When $\tau$ is clear from context, we say
$\alpha$ is \emph{stable}.  Given two assemblies $\alpha,\beta$, we
say $\alpha$ is a \emph{subassembly} of $\beta$, and we write $\alpha
\sqsubseteq \beta$, if $\dom\alpha \subseteq \dom\beta$ and, for all
points $\vec{p} \in \dom\alpha$, $\alpha(\vec{p}) = \beta(\vec{p})$.

A $d$-dimensional \emph{tile assembly system} (TAS) is a triple $\mathcal{T} =
(T,\sigma,\tau)$, where $T$ is a tile set, $\sigma:\mathbb{Z}^d
\dashrightarrow T$ is the finite, $\tau$-stable, \emph{seed assembly},
and $\tau\in\mathbb{Z}^+$ is the \emph{temperature}.

Given two $\tau$-stable assemblies $\alpha,\beta$, we write $\alpha
\to_1^{\mathcal{T}} \beta$ if $\alpha \sqsubseteq \beta$ and
$|\dom\beta \setminus \dom\alpha| = 1$. In this case we say $\alpha$
\emph{$\mathcal{T}$-produces $\beta$ in one step}. If $\alpha
\to_1^{\mathcal{T}} \beta$, $ \dom\beta \setminus
\dom\alpha=\{\vec{p}\}$, and $t=\beta(\vec{p})$, we write $\beta =
\alpha + (\vec{p} \mapsto t)$.  The \emph{$\mathcal{T}$-frontier} of
$\alpha$ is the set $\partial^\mathcal{T} \alpha = \bigcup_{\alpha
  \to_1^\mathcal{T} \beta} (\dom\beta \setminus \dom\alpha$), i.e.,
the set of empty locations at which a tile could stably attach to
$\alpha$. The \emph{$t$-frontier} of $\alpha$, denoted
$\partial^\mathcal{T}_t \alpha$, is the subset of
$\partial^\mathcal{T} \alpha$ defined as
$\setr{\vec{p}\in\partial^\mathcal{T} \alpha}{\alpha \to_1^\mathcal{T}
  \beta \text{ and } \beta(\vec{p})=t}.$

Let $\mathcal{A}^T$ denote the set of all assemblies of tiles from $T$, and let $\mathcal{A}^T_{< \infty}$ denote the set of finite assemblies of tiles from $T$.
A sequence of $k\in\mathbb{Z}^+ \cup \{\infty\}$ assemblies $\vec{\alpha} = \left (\alpha_0,\alpha_1,\ldots \right)$ over $\mathcal{A}^T$ is a \emph{$\mathcal{T}$-assembly sequence} if, for all $1 \leq i < k$, $\alpha_{i-1} \to_1^\mathcal{T} \alpha_{i}$.
The {\em result} of an assembly sequence $\vec{\alpha}$, denoted as $\textmd{res}(\vec{\alpha})$, is the unique limiting assembly (for a finite sequence, this is the final assembly in the sequence).

We write $\alpha \to^\mathcal{T} \beta$, and we say $\alpha$ \emph{$\mathcal{T}$-produces} $\beta$ (in 0 or more steps), if there is a $\mathcal{T}$-assembly sequence $\alpha_0,\alpha_1,\ldots$ of length $k = |\dom\beta \setminus \dom\alpha| + 1$ such that
(1) $\alpha = \alpha_0$,
(2) $\dom\beta = \bigcup_{0 \leq i < k} \dom{\alpha_i}$, and
(3) for all $0 \leq i < k$, $\alpha_{i} \sqsubseteq \beta$.
If $k$ is finite then it is routine to verify that $\beta = \alpha_{k-1}$.

We say $\alpha$ is \emph{$\mathcal{T}$-producible} if $\sigma \to^\mathcal{T} \alpha$, and we write $\prodasm{\mathcal{T}}$ to denote the set of $\mathcal{T}$-producible assemblies. 

An assembly $\alpha$ is \emph{$\mathcal{T}$-terminal} if $\alpha$ is $\tau$-stable and $\partial^\mathcal{T} \alpha=\emptyset$. 
We write $\termasm{\mathcal{T}} \subseteq \prodasm{\mathcal{T}}$ to denote the set of $\mathcal{T}$-producible, $\mathcal{T}$-terminal assemblies. If $|\termasm{\mathcal{T}}| = 1$ then  $\mathcal{T}$ is said to be {\em directed}.

In general, a $d$-dimensional shape is a set $X \subseteq \mathbb{Z}^d$. We say that a TAS $\mathcal{T}$ \emph{self-assembles} $X$ if, for all $\alpha \in
\termasm{\mathcal{T}}$, $\dom\alpha = X$, i.e., if every terminal
assembly produced by $\mathcal{T}$ places a tile on every point in $X$
and does not place any tiles on points in $\mathbb{Z}^d \backslash\, X$. We say that a TAS $\mathcal{T}$ \emph{uniquely 
  self-assembles} $X$ if $\termasm{\mathcal{T}} = \{ \alpha \}$ and $\dom\alpha = X$.

In the spirit of \cite{RotWin00}, we define the \emph{tile complexity} of a shape $X$ at temperature $\tau$, denoted by $K^\tau_{SA}(X)$, as the minimum number of distinct tile types of any TAS in which it self-assembles, i.e., \\$K^\tau_{SA}(X) = \min \left\{ n   \; \left| \; \mathcal{T} = \left(T,\sigma,\tau\right), \left|T\right|=n \textmd{ and } X \textmd{ self-assembles in } \mathcal{T} \right.\right\}$. The \emph{directed tile complexity} of a shape $X$ at temperature $\tau$, denoted by $K^\tau_{USA}(X)$, is the minimum number of distinct tile types of any TAS in which it uniquely self-assembles, i.e., \\$K^\tau_{USA}(X) = \min \left\{ n   \; \left| \; \mathcal{T} = \left(T,\sigma,\tau\right), \left|T\right|=n \textmd{ and } X \textmd{ uniquely self-assembles in } \mathcal{T} \right.\right\}$.

\section{Lower bound}
\label{sec:impossibility}

In this section, we prove Theorem~\ref{thm:one}, which is a lower bound on the tile complexity of 2D rectangles. So, going forward, let $k,N \in \mathbb{N}$. We say that $R^2_{k,N}$ is a 2D $k \times N$ \emph{rectangle} if $R^2_{k,N} = \{0,1, \ldots, N-1\} \times \{0,1,\ldots, k-1\}$. Throughout this section, we will denote $R^2_{k,N}$ as simply $R_{k,N}$. Our lower bound relies on the following observation regarding temperature-1 self-assembly.

\begin{observation}
\label{obs:simple}
If $\mathcal{T} = (T,\sigma,1)$ is a singly-seeded TAS in which some shape $X$ self-assembles and $\alpha \in
\mathcal{A}[\mathcal{T}]$ such that $\mathrm{dom}\ \alpha = X$, then $G^{\textmd{b}}_{\alpha}$ contains a simple path $s$
from the location of $\sigma$ to any location of $X$ and there is a corresponding (simple) assembly sequence $\vec{\alpha}$ that follows $s$ by placing tiles on and only on locations in $s$.
\end{observation}

Since, in Observation~\ref{obs:simple}, we do not necessarily assume that $\mathcal{T}$ uniquely produces $X$, there could be more than one assembly sequence for a given $s$. Throughout the rest of this section, unless stated otherwise, let $\mathcal{T} = (T,\sigma,1)$ be a singly-seeded TAS in which $R_{k,N}$ self-assembles.

\subsection{Window Movie Lemmas}

To prove our lower bound, we will use a variation of the Window Movie Lemma (WML) by Meunier, Patitz, Summers, Theyssier, Winslow and Woods and a corollary thereof. In this subsection, we review standard notation \cite{WindowMovieLemma} for and give the statements of the variation that we use in our lower bound proof.

A \emph{window} $w$ is a set of edges forming a cut-set of the full grid
graph of $\mathbb{Z}^d$. Given a window $w$ and an assembly $\alpha$, a window that {\em
  intersects} $\alpha$ is a partitioning of $\alpha$ into two
configurations (i.e., after being split into two parts, each part may
or may not be disconnected). In this case we say that the window $w$
cuts the assembly $\alpha$ into two configurations $\alpha_L$
and~$\alpha_R$, where $\alpha = \alpha_L \cup \alpha_R$. Given a window $w$, its translation by a vector $\vec{\Delta}$, written $ w
+ \vec{\Delta}$ is simply the translation of each one of $w$'s elements (edges)
by~$\vec{\Delta}$.

For a window $w$ and an assembly sequence $\vec{\alpha}$, we define a \emph{glue window movie}~$M$ to be the order of placement, position and glue type for each glue that appears along the window $w$ in $\vec{\alpha}$. Given an assembly sequence $\vec{\alpha}$ and a window $w$, the associated glue window movie is the maximal sequence $M_{\vec{\alpha},w} =  \left(\vec{v}_{1}, g_{1}\right), \left(\vec{v}_{2}, g_{2}\right), \ldots$ of pairs of grid graph vertices $\vec{v}_i$ and glues $g_i$, given by the order of the appearance of the glues along window $w$ in the assembly sequence $\vec{\alpha}$.
Furthermore, if $m$ glues appear along $w$ at the same instant (this happens upon placement of a tile which has multiple sides touching $w$) then these $m$ glues appear contiguously and are listed in lexicographical order of the unit vectors describing their orientation in $M_{\vec{\alpha},w}$. We use the notation $\mathcal{B}\left(M_{\vec{\alpha}, w}\right)$ to denote the \emph{bond-forming submovie} of $M_{\vec{\alpha},w}$, which consists of only those steps of $M_{\vec{\alpha},w}$ that place glues that eventually form positive-strength bonds in the assembly $\textmd{res}\left(\vec{\alpha}\right)$. We write $M_{\vec{\alpha}, w} + \vec{\Delta}$ to denote the translation of the glue window movie, that is $\left(\vec{v}_{1}+\vec{\Delta}, g_{1}\right), \left(\vec{v}_{2}+\vec{\Delta}, g_{2}\right), \ldots$.

Let $w$ be a window that partitions~$\alpha$ into two configurations~$\alpha_L$ and $\alpha_R$, and assume $w + \vec{\Delta}$ partitions~$\beta$ into two configurations $\beta_L$ and $\beta_R$.
Assume that $\alpha_L$, $\beta_L$ are the sub-configurations of $\alpha$ and $\beta$ containing the seed tile of $\alpha$ and $\beta$, respectively. 

\begin{lemma}[Standard Window Movie Lemma]\label{lem:wml}
If $\mathcal{B}\left(M_{\vec{\alpha},w}\right) = \mathcal{B}\left(M_{\vec{\beta},w+\vec{\Delta}}\right) - \vec{\Delta}$, then the following two assemblies are producible: $\alpha_L \left(\beta_R - \vec{\Delta}\right) = \alpha_L \cup \left(\beta_R - \vec{\Delta}\right)$ and $\beta_L\left(\alpha_R + \vec{\Delta}\right) = \beta_L \cup \left(\alpha_R + \vec{\Delta}\right)$.
\end{lemma}

By Observation~\ref{obs:simple}, there exist simple assembly sequences $\vec{\alpha} = \left(\alpha_i \mid 0 \leq i < l\right)$ and $\vec{\beta} = \left(\beta_i \mid 0 \leq i < m\right)$, with $l,m\in\mathbb{Z}^+ $ that place tiles along simple paths $s$ and $s'$, respectively, leading to results $\alpha$ and $\beta$, respectively. The notation $M_{\vec{\alpha}, w} \upharpoonright s$ represents the \emph{restricted} glue window submovie (\emph{restricted to} $s$), which consists of only those steps of $M$ that place glues that eventually form positive-strength bonds along $s$.

\begin{corollary}[Restricted Window Movie Lemma]\label{lem:rwml}
If\\$M_{\vec{\alpha},w} \upharpoonright s  = \left( M_{\vec{\beta},w+\vec{\Delta}} \upharpoonright s' \right) - \vec{\Delta}$, then the following two assemblies are producible: $\alpha_L \left(\beta_R - \vec{\Delta}\right) = \alpha_L \cup \left(\beta_R - \vec{\Delta}\right)$ and $\beta_L\left(\alpha_R + \vec{\Delta}\right) = \beta_L \cup \left(\alpha_R + \vec{\Delta}\right)$.
\end{corollary}

The proof of Corollary~\ref{lem:rwml} is identical to that of Lemma~\ref{lem:wml} and therefore is omitted. The proof is identical because $\vec{\alpha}$ $\left(\vec{\beta}\right)$ follows $s$ ($s'$) and $\tau=1$, so glues that do not follow $s$ ($s'$) are not necessary for the self-assembly of $\alpha$ ($\beta$). Note that Lemma~\ref{lem:wml} and Corollary~\ref{lem:rwml} both generalize to 3D. See Figure~\ref{fig:glue_movie_window_examples} for examples of $M_{\vec{\alpha},w}$, $\mathcal{B}\left( M_{\vec{\alpha},w} \right)$ and $M_{\vec{\alpha}, w} \upharpoonright s$. We now turn our attention to counting the number of restricted glue window submovies.

\begin{figure}
    \centering

    \begin{subfigure}[t]{0.15\textwidth}
        \centering
        \includegraphics[width=.75in]{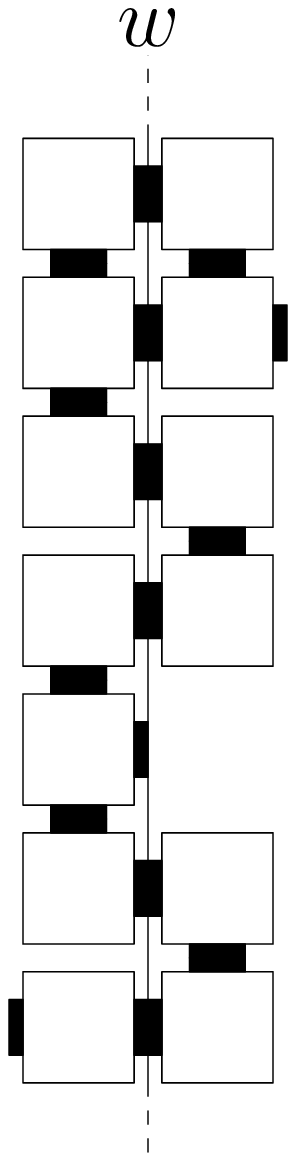}
        \caption{\label{fig:glue_movie_window_example_assembly} A subassembly of $\alpha$ and a window $w$ induced by a translation of the $y$-axis.}
    \end{subfigure}%
    ~ \hspace{15pt}
    \begin{subfigure}[t]{0.15\textwidth}
        \centering
        \includegraphics[width=.75in]{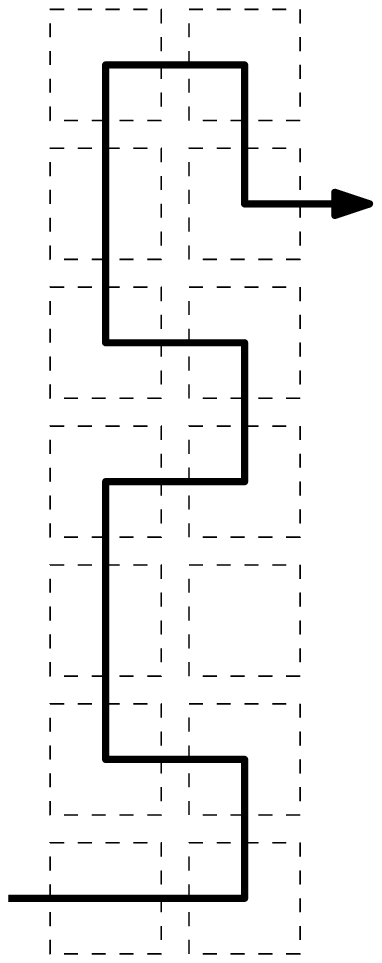}
        \caption{\label{fig:glue_movie_window_example_just_path} A portion of the simple path $s$ through $G^b_{\alpha}$.}
    \end{subfigure}%
   ~ \hspace{15pt}
    \begin{subfigure}[t]{0.15\textwidth}
        \centering
        \includegraphics[width=.75in]{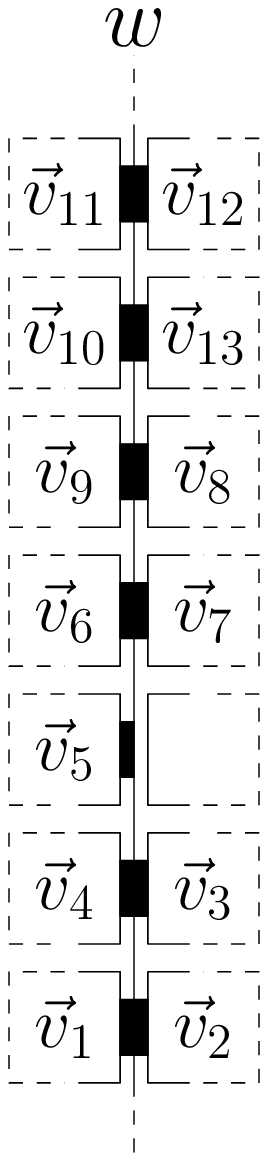}

        \caption{\label{fig:glue_movie_window_example_full} The glue window movie $M_{\vec{\alpha},w}$ }
    \end{subfigure}%
    ~ \hspace{15pt}
    \begin{subfigure}[t]{0.15\textwidth}
        \centering
        \includegraphics[width=.75in]{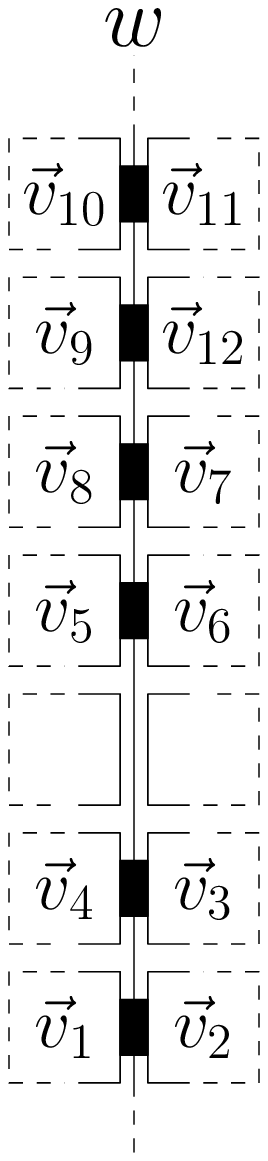}
        \caption{\label{fig:glue_movie_window_example_bond_forming_submovie} The bond-forming submovie $\mathcal{B}\left( M_{\vec{\alpha},w} \right)$ }
    \end{subfigure}%
    ~ \hspace{15pt}
    \begin{subfigure}[t]{0.15\textwidth}
        \centering
        \includegraphics[width=.75in]{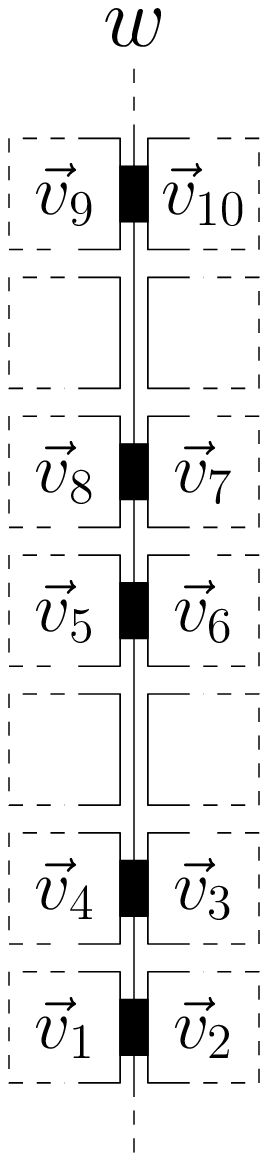}
        \caption{\label{fig:glue_movie_window_example_restricted} The restricted glue window submovie $M_{\vec{\alpha},w} \upharpoonright s$ }
    \end{subfigure}
    \caption{\label{fig:glue_movie_window_examples} An assembly, a simple path and the various types of glue window movies. }
    
\end{figure}

\subsection{Counting procedure for undirected self-assembly in 2D}
\label{sec:counting_procedure}

In this subsection, we develop a counting procedure that we will use to obtain an upper bound on the number of distinct restricted glue window submovies. 

Let $\rho \in \mathcal{A}_\Box\left[ \mathcal{T} \right]$. By Observation~\ref{obs:simple}, there exists a simple path $s$ in $G^{\textmd{b}}_{\rho}$
from the location of $\sigma$ to any location in an extreme (i.e., leftmost or rightmost) column
of $R_{k,N}$. For the remainder of this section, unless stated otherwise, let $\vec{\alpha}$ denote any (simple) assembly sequence that follows a simple path from the location of $\sigma$ to some location in the furthest extreme column of $R_{k,N}$. So, if $\sigma$ is in the left half of $R_{k,N}$, then $s$ will go from $\sigma$ to some location in the rightmost column of $R_{k,N}$. If $\sigma$ is in the right half of $R_{k,N}$, then $s$ will go from $\sigma$ to some location in the leftmost column of $R_{k,N}$. If $\sigma$ is in the middle column of $R_{k,N}$, then $s$ can go to either the leftmost or rightmost column of $R_{k,N}$. 

Index the columns of $R_{k,N}$ from 1 (left) to $N$ (right). Assume $c_\sigma$ is the index of the column in which the seed is contained. 
Consider two consecutive columns, with indices $c_0$ and $c_1$, satisfying $\left| c_\sigma - c_0 \right| < \left| c_\sigma - c_1 \right|$, and such that either $c_0$ or $c_1$ (or both) are in between $c_{\sigma}$ and the furthest (from $\sigma$) extreme column. 
Since $s$ is a simple path from the location of $\sigma$ to some location in the extreme column of $R_{k,N}$ that is furthest from $\sigma$, $s$ crosses between $c_0$ and $c_1$ through $e$ \emph{crossing edges} in $G^{\textmd{b}}_{\rho}$, where $1 \leq e \leq k$ and $e$ is odd, visiting a total of $2e$ endpoints.  
The endpoint of a crossing edge in column $c_0$ ($c_1$) is its \emph{near} (far) endpoint. A crossing edge points \emph{away from} (towards) the seed if its near endpoint is visited first (second).

Observe that the first and last crossing edges visited by $s$ must point away from the seed but each crossing edge that points away from the seed (except the last crossing edge) is immediately followed by a corresponding crossing edge that points towards the seed, when skipping the part of $s$ that connects them without going through another crossing edge in between them. Assume that the rows of $R_{k,N}$ are assigned an index from $1$ (top) to $k$ (bottom). Let $E \subseteq \{1,\ldots, k\}$ be such that $|E| = e$ and $f \in E$ and $l \in E$ be the row indices of the first and last crossing edges visited, respectively. We define a \emph{near} (far) \emph{crossing pairing over} $E$ \emph{starting at} $f \in E$ (ending at $l \in E$) as a set of $p = \frac{e-1}{2}$ non-overlapping pairs of elements in $E\backslash \{f\}$ ($E\backslash\{l\}$), where each pair contains (the row indices of) one crossing edge pointing away from the seed and its corresponding crossing edge pointing towards the seed. See Figure~\ref{fig:counting_proc_good} for examples of the previous definitions. For counting purposes pertaining to a forthcoming argument, we establish an injective mapping
from near crossing pairings over $E$ to strings of balanced parentheses of length $e - 1$. 

\begin{lemma}
\label{lem:near-crossing-pairing-to-parens}
There exists an injective function from the set of all near crossing pairings over $E$ starting at $f$ into the set of all strings of $2p$ balanced parentheses.
\end{lemma}

\begin{proof}
Given a near crossing pairing $P$ with $p$ pairs, build a string $x$ with $2p$
characters indexed from 1 to $2p$ going from left to right, as
follows. For each element $\{a,b\}$ of $P$, with $a<b$ and assuming $a$ is the $i$-th lowest row index and $b$ is the $j$-th lowest row index in $E \backslash \{f\}$, place a left parenthesis at index $i$ and a right parenthesis at index $j$.

The resulting string $x$ contains
exactly $p$ pairs of parentheses. Furthermore, all of the parentheses
in $x$ are balanced because each opening parenthesis appears to the
left of its closing parenthesis (thanks to the indexing used in the
string construction algorithm just described) and, for any two pairs
of parentheses in the string, it must be the case that either 1) they do not
overlap (i.e., the closing parenthesis of the leftmost pair is
positioned to the left of the opening parenthesis of the rightmost
pair) or 2) one pair is nested inside the other (i.e., the interval
defined by the indices of the nested pair is included in the interval
defined by the indices of the outer pair). The other case is
impossible, that is, when the two pairs $\{a,b\}$ and $\{c,d\}$ are
such that $a < c < b < d$ because it would be impossible for any path
crossing consecutive columns according to $P$ to be simple. Consider any simple
path $\pi_{\{a,b\}}$ that links crossing edges $a$ and $b$ without going through another crossing edge in between them. Since $P$ is a near crossing pairing, $\pi_{\{a,b\}}$ is fully
contained in the half-plane $H_0$ on the near side of $c_0$ toward the seed. This
path partitions $H_0$ into two spaces. Since the crossing edges $c$ and $d$ belong to different components of this
partition, any simple path linking these two crossing edges must
cross $\pi_{\{a,b\}}$. Therefore, $s$, crossing between $c_0$ and $c_1$, could not have been simple.

Finally, note that two different near crossing pairings $P_1$ and $P_2$ over $E$ starting at $f$ will map to two different strings of balanced parentheses, so the mapping is injective.
\qed
\end{proof}

\begin{corollary}
\label{lem:far-crossing-pairing-to-parens}
There exists an injective function from the set of all far crossing pairings over $E$ ending at $l$ into the set of all strings of $2p$ balanced parentheses.
\end{corollary}

\begin{figure}
    \centering

    \begin{subfigure}[t]{0.48\textwidth}
        \centering
        \includegraphics[width=2.2in]{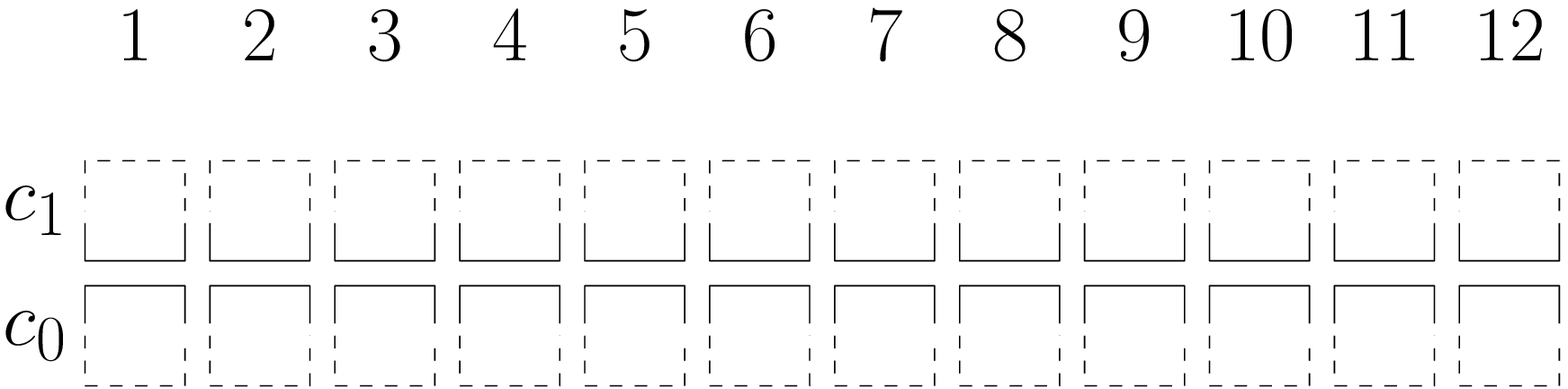}
        \caption{\label{fig:counting_proc_0}Two consecutive columns, for ``height'' $k=12$, rotated 90 degrees counter-clockwise. We show the columns this way for ease of examining the corresponding strings of balanced parentheses.}
    \end{subfigure}%
    ~
    \begin{subfigure}[t]{0.48\textwidth}
        \centering
        \includegraphics[width=2.2in]{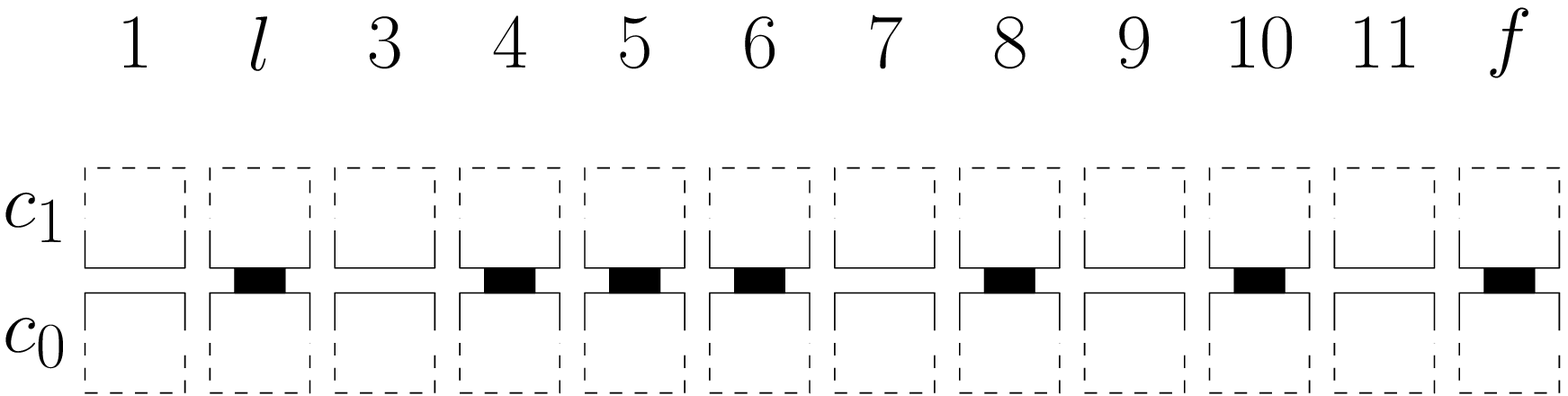}
        \caption{\label{fig:counting_proc_steps_1_2_3}Steps 1, 2 and 3. Here, $f=12$, $l=2$ and $E = \{2, 4, 5, 6, 8, 10, 12\}$, with chosen glues indicated by the little black rectangles.}
    \end{subfigure}%
    \\
   \vspace{10pt}
    \begin{subfigure}[t]{0.48\textwidth}
        \centering
        \includegraphics[width=2.2in]{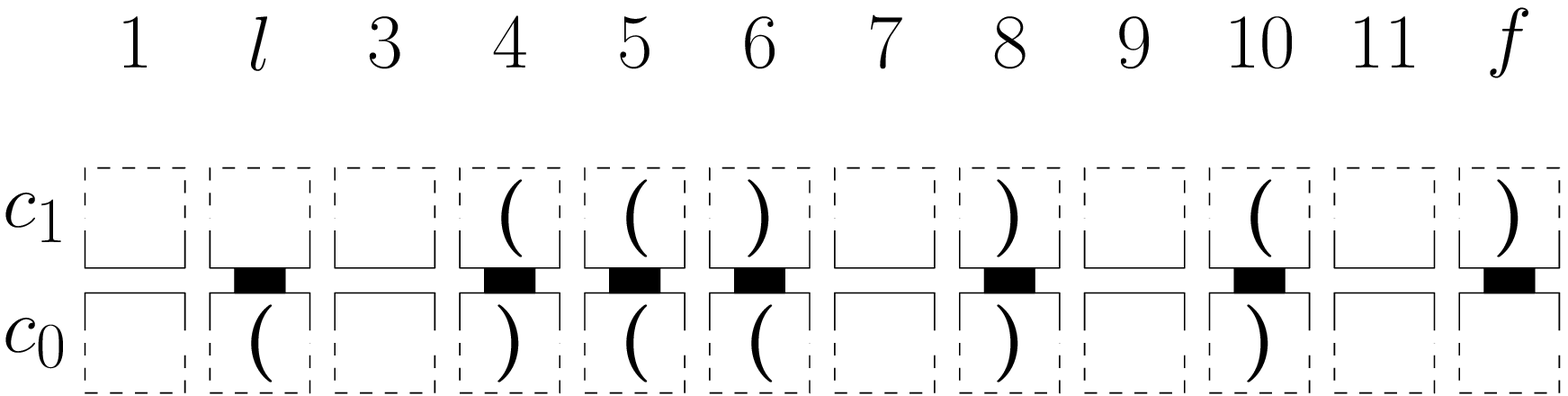}

        \caption{\label{fig:counting_proc_steps_4_5}Steps 4 and 5. Here, $x_0 = ()(())$ and $x_1 = (())()$.}
    \end{subfigure}%
    ~
    \begin{subfigure}[t]{0.48\textwidth}
        \centering
        \includegraphics[width=2.2in]{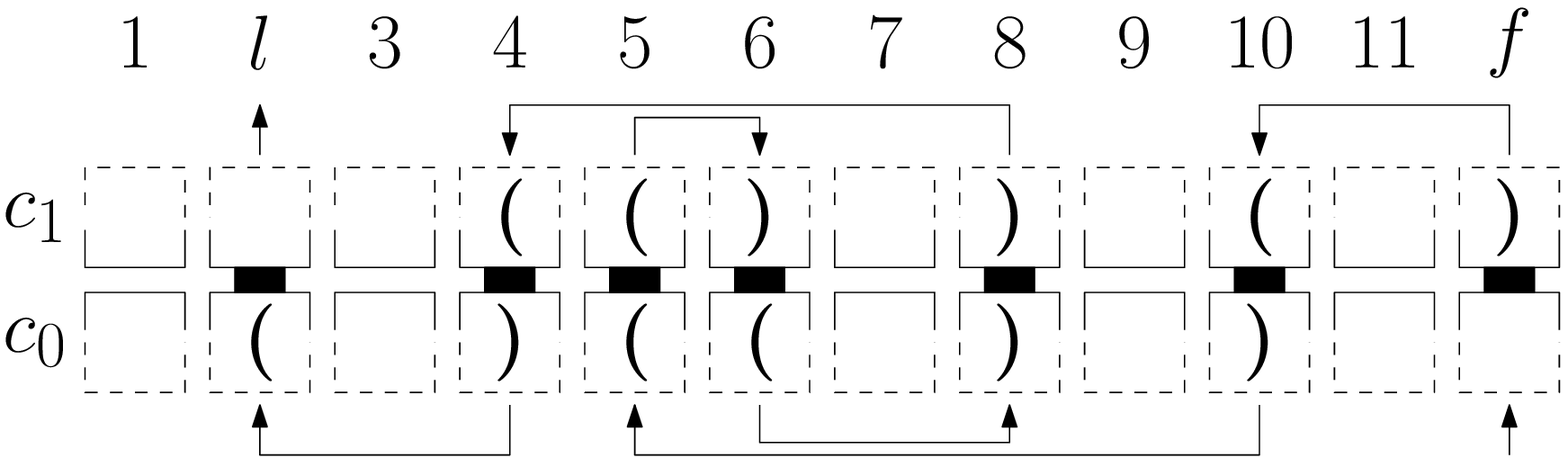}
        \caption{\label{fig:counting_proc_step_6}The $\pi$ with $|\pi|=14=2e$ that is found by the algorithm in Step 6. }
    \end{subfigure}
    \caption{\label{fig:counting_proc_good}A ``good'' sample run of the counting procedure.}
    
    \vspace{-20pt}
\end{figure}

The following is a systematic
procedure for upper bounding the number of ways to select and order the $2e$ endpoints of the $e$ crossing edges between an arbitrary pair of consecutive columns $c_0$ and $c_1$ (see Figure~\ref{fig:counting_proc_0}). Obviously, the number of ways to do this is less than or equal to $\binom{k}{e} (2e)!$. We will use crossing pairings and Catalan numbers to reduce this upper bound.

\begin{enumerate}
		\item  Choose the set $E$ of row indices of the $e$ crossing edges, out of $k$ possible edges between consecutive columns. There are $\binom{k}{e}$ ways to do this.

    \item One of the crossing edges must be first, so, choose the first crossing edge $f$. There are $e$ ways to do this.

        \item One of the crossing edges must be last, so, choose the last crossing edge $l$. There are $e$ ways to do this.
        
\end{enumerate}

The three previous steps are depicted in Figure~\ref{fig:counting_proc_steps_1_2_3}. We purposely allow choosing $l=f$ because our intention is to upper bound the number of ways to select and order the endpoints of the crossing edges. Moreover, $l=f$ when $e=1$. Denote as $\mathcal{C}_p = \frac{1}{p+1} \binom{2p}{p}$ the $p^{th}$ Catalan number (indexed starting at 0).

\begin{enumerate}

\setcounter{enumi}{3}

		\item For a given pair of consecutive columns, in which $f$ is visited first, $s$ induces a near crossing pairing over $E$ starting at $f$, where the elements of each pair are row indices of near endpoints of crossing edges, including $l$, but not $f$. By Lemma~\ref{lem:near-crossing-pairing-to-parens}, it suffices to count the number of ways to choose a string of $2p$ balanced parentheses. Therefore, choose a string $x_0$ of $2p$ balanced parentheses, where $x_0[i]$ corresponds to the crossing edge in $E$ with the $i$-th lowest row index, excluding $f$. There are $\mathcal{C}_p$ ways to do this.

        \item For a given pair of consecutive columns in which $l$ is visited last, $s$ induces a far crossing pairing over $E$ ending at $l$, where the elements of each pair are row indices of far endpoints of crossing edges, including $f$, but not $l$. By Corollary~\ref{lem:far-crossing-pairing-to-parens}, it suffices to count the number of ways to choose a string of $2p$ balanced parentheses. Therefore, choose a string $x_1$ of $2p$ balanced parentheses, where $x_1[i]$ corresponds to the crossing edge in $E$ with the $i$-th lowest row index, excluding $l$. There are $\mathcal{C}_p$ ways to do this.
        
        \end{enumerate}

The two previous steps are depicted in Figure~\ref{fig:counting_proc_steps_4_5}. At this point, we have chosen the locations and connectivity pattern of all the crossing edges. We now show that there is at most one way in which both endpoints of every crossing edge may be visited by $s$, subject to the  constraints imposed by the previous steps.

\begin{enumerate}

\setcounter{enumi}{5}

\item Let $I_j(r)$ be the index $i$, such that $x_j[i]$ corresponds to the crossing edge with row index $r$. The following greedy algorithm attempts to build a path $\pi$ of locations that (1) starts at the near endpoint of the first crossing edge, (2) ends at the far endpoint of the last crossing edge and (3) visits only the endpoints of the crossing edges while following the balanced parenthesis pairings of both $x_0$ and $x_1$.

    \begin{algorithm}[H]
    	\SetAlgoLined
    	Initialize $\pi = \left( \left(c_0, f\right), \left(c_1, f\right) \right)$ to be a sequence of locations, $j = 1$ and $r = f$, where  $r$ stands for ``row number'' and $j$ stands for the current side, near (0) or far (1).

        \While{$r\ne l$}{
            Let $r'$ be the unique row index of the crossing edge, such that, $I_j(r)$ and $I_j\left(r'\right)$ are paired.

        	Set $r = r'$.

        	Append $\left(c_j, r\right)$ to $\pi$.

            Set $j = (j+1) \mod 2$.

            Append $\left(c_j, r\right)$ to $\pi$.

        }
    \end{algorithm}

	First, note that no endpoint can be visited more than once, so the algorithm always terminates. Second, note that, when the algorithm terminates, either $|\pi| = 2e$ (see Figure~\ref{fig:counting_proc_step_6}) or $|\pi|<2e$ (see Figure~\ref{fig:counting_proc_bad}). However, regardless of its length, $\pi$ is the unique sequence of endpoints of crossing edges that can be visited by any simple path starting at $\left(c_0,f\right)$, ending at $\left(c_1,l\right)$ and following the balanced parenthesis pairings of both $x_0$ and $x_1$. The uniqueness of $\pi$ follows from the uniqueness of $r'$, given $r$, based on the balanced parenthesis pairings of both $x_0$ and $x_1$.

\begin{figure}[h]
    \centering \includegraphics[width=2.2in]{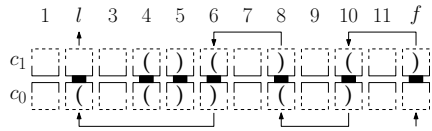}
        \caption{\label{fig:counting_proc_bad}A ``bad'' sample run of the counting procedure. Here, $x_0=(())()$ and $x_1=()()()$. With $f=12$, $l=2$ and $E = \{2, 4, 5, 6, 8, 10, 12\}$, there is no valid $\pi$ (crossing edges 4 and 5 are skipped, assuming the pairs of parentheses are faithfully followed), so the algorithm terminates with $|\pi|=10<2e=14$.}

	\vspace{-20pt}
    \end{figure}

\end{enumerate}

By the above counting procedure, there are at most $\binom{k}{e}\left( e\frac{1}{p+1}\binom{2p}{p}\right)^2\cdot 1$ ways to select and order the endpoints of the crossing edges between an arbitrary pair of consecutive columns $c_0$ and $c_1$ as they are visited by a simple path.

\subsection{Lower bound for undirected self-assembly in 2D: Theorem~\ref{thm:one}}

To prove a lower bound on $K^1_{SA}\left(R_{k,N}\right)$, we turn our attention to upper bounding the number of restricted glue window submovies of the form $M_{\vec{\alpha}, w} \upharpoonright s$. For the remainder of this subsection, assume $w$ is always some window induced by (a translation of) the $y$-axis that cuts $R_{k,N}$ between some pair of consecutive columns. 

\begin{lemma}
\label{lem:count_restricted_bond_forming_submovies}
The number of restricted glue window submovies of the form $M_{\vec{\alpha},w}~\upharpoonright~s$ is less than or equal to $|G|^k \cdot 2^{3k+2}  \cdot k$, where $G$ is the set of all glues of (the tile types in) $T$.
\end{lemma}

\begin{proof}
Let $e$ be an odd number such that $1 \leq e \leq k$ and\\$M_{\vec{\alpha},w} \upharpoonright s = \left(\vec{v}_1,g_1\right), \ldots, \left(\vec{v}_{2e}, g_{2e}\right)$ be a restricted glue window submovie. Since $\vec{\alpha}$ follows a simple path, $g_{2i-1}=g_{2i}$ for $i=1,\ldots, e$. This means that we only need to assign $e$ glues, with $|G|$ choices for each glue. So, the number of ways to assign glues in $M_{\vec{\alpha},w} \upharpoonright s$ is less than or equal to $|G|^e$. Since $\vec{\alpha}$ follows a simple path, each location in $M_{\vec{\alpha},w} \upharpoonright s$ corresponds to an endpoint of a crossing edge that crosses $w$. So, the number of ways to assign locations in $M_{\vec{\alpha},w} \upharpoonright s$ is less than or equal to the number of ways to select and order the endpoints of $e$ crossing edges that cross $w$ via $\vec{\alpha}$. By the above counting procedure, the number of ways to select and order the endpoints of $e$ crossing edges that cross $w$ via $\vec{\alpha}$ is less than or equal to $\binom{k}{e}\left( e\frac{1}{p+1}\binom{2p}{p}\right)^2$. Thus, if $m$ is the total number of restricted glue window submovies of the form $M_{\vec{\alpha},w} \upharpoonright s$, then we have:
\begin{eqnarray*}
    m & \leq &  \sum_{\substack{1\leq e \leq k \\ e\ \mathrm{odd}}}\left( \binom{k}{e}\left( e\frac{1}{p+1}\binom{2p}{p}\right)^2  \ |G|^e \right) \\
         	& = & \sum_{\substack{1\leq e \leq k \\ e\ \mathrm{odd}}}\left( \binom{k}{e}\left( e\frac{2}{e+1}\binom{e-1}{(e-1)/2}\right)^2  \ |G|^e \right) \\
     &\leq & 2^2 \ \sum_{\substack{1\leq e \leq k \\ e\ \mathrm{odd}}}\left( \binom{k}{e}\left( \frac{e}{e+1}\binom{e-1}{(e-1)/2}\right)^2  \ |G|^k \right)  \\
    & = &  |G|^k \cdot 2^2 \  \sum_{\substack{1\leq e \leq k \\ e\ \mathrm{odd}}}\left( \binom{k}{e}\left( \frac{e}{e+1}\binom{e-1}{(e-1)/2}\right)^2 \right) \\
     &\leq &  |G|^k \cdot 2^2 \ \sum_{\substack{1\leq e \leq k\\ e\ \mathrm{odd}}}\left( 2^k\left( 1 \cdot \binom{e-1}{(e-1)/2}\right)^2 \right) \\
     & = &   |G|^k \cdot 2^{k+2} \ \sum_{\substack{1\leq e \leq k\\ e\ \mathrm{odd}}}  \binom{e-1}{(e-1)/2}^2 \\
     & = & |G|^k \cdot 2^{k+2} \ \sum_{\substack{0\leq e \leq k-1\\ e\ \mathrm{even}}}  \binom{e}{e/2}^2 \\
     &  \leq &   |G|^k \cdot 2^{k+2} \ \sum_{\substack{0\leq e \leq k-1\\ e\ \mathrm{even}}}  \binom{k}{e/2}^2  \\
     &  \leq & |G|^k  \cdot 2^{k+2}\  \ \sum_{\substack{0\leq e \leq k-1\\ e\ \mathrm{even}}}2^{2k} \\
     & \leq & |G|^k  \cdot 2^{3k+2} \cdot k. \\
\end{eqnarray*}
\qed
\end{proof}

We will use Lemma~\ref{lem:count_restricted_bond_forming_submovies} to prove our impossibility result.

\addtocounter{theorem}{-2}

\begin{theorem}
  $K^1_{SA}\left(R_{k,N}\right)=\Omega\left( N^{\frac{1}{k}}\right)$.
\end{theorem}

\begin{proof}

Let $G$ be the set of glues of (the tile types in) $T$. It suffices to show that $|T| = \Omega\left( N^{\frac{1}{k}}\right)$. Since $s$ is the longest path in $G^{\textmd{b}}_{\rho}$, from the location of $\sigma$ to some location in an extreme column
of $R_{k,N}$, by Lemma~\ref{lem:count_restricted_bond_forming_submovies}, if $N > 2 \cdot |G|^k  \cdot 2^{3k+2}  \cdot k$, then there exists a window $w$, along with vectors $\vec{\Delta}_1$ and $\vec{\Delta}_2$, such that, $\vec{\Delta}_1 \ne \vec{0}$,  $\vec{\Delta}_2 \ne \vec{0}$ and $\vec{\Delta}_1 \ne \vec{\Delta}_2$ and satisfying $M_{\vec{\alpha},w} \upharpoonright s = \left(M_{\vec{\alpha},w+\vec{\Delta}_1} \upharpoonright s \right) - \vec{\Delta}_1$ and $M_{\vec{\alpha},w+\vec{\Delta}_1} \upharpoonright s = \left(M_{\vec{\alpha},w+\vec{\Delta}_2} \upharpoonright s \right) - \vec{\Delta}_2$. Without loss of generality, assume that $w$ and $w+\vec{\Delta}_1$ are on the same side of $\sigma$ and $w+\vec{\Delta}_1$ is to the left of $w$. Assume that $w$ partitions $\alpha$ into $\alpha_L$ and $\alpha_R$ and $w+\vec{\Delta}_1$ partitions $\beta=\alpha$ into $\beta_L$ and $\beta_R$. Then, by Corollary~\ref{lem:rwml}, $\beta_L\left(\alpha_R+\vec{\Delta}_1\right) \in \mathcal{A}[\mathcal{T}]$.  Since $\vec{\Delta}_1\ne \vec{0}$, $\textmd{dom}\left(\beta_L\left(\alpha_R + \vec{\Delta}_1\right)\right)\backslash R_{k,N} \ne \emptyset$. In other words, $\mathcal{T}$ produces some assembly that places at least one tile at a location that is not an element of $R_{k,N}$. Therefore, it must hold that  $N \leq 2 \cdot |G|^k \cdot 2^{3k+2}  \cdot k$, which implies that $|G|^k \geq \frac{N}{ 2^{3k+3}  \cdot k}$, and thus $|G| \geq \left( \frac{N}{2^{3k+3}  \cdot k} \right)^{\frac{1}{k}} \geq \frac{N^{\frac{1}{k}}}{\left( 2^{6k}  \cdot 2^k \right)^{\frac{1}{k}}} = \frac{N^{\frac{1}{k}}}{128}$. Finally, note that $|T| \geq \frac{|G|}{4}$ and it follows that $|T| = \Omega\left( N^{\frac{1}{k}} \right)$.
\qed

\end{proof}

Note that the main technique for the proof of Theorem~\ref{thm:one} does not hold in 3D because Lemma~\ref{lem:near-crossing-pairing-to-parens} requires planarity.

%
%
%
%
%
%
%
\section{Upper bound}
\label{sec:thin-rectangle-construction}

In this section, we give a construction for a singly-seeded TAS in which a sufficiently large just-barely 3D rectangle uniquely self-assembles. Going forward, we say that $R^3_{k,N} \subseteq \mathbb{Z}^3$ is a 3D $k \times N$ \emph{rectangle} if $\{0,1, \ldots,N-1\} \times \{0,1,\ldots,k-1\} \times \{0\} \subseteq R^3_{k,N} \subseteq \{0,1\ldots,N-1\} \times \{0,1\ldots,k
-1\} \times \{0,1\}$. For the sake of clarity of presentation, we represent $R^3_{k,N}$ vertically. Here is our main positive result. 

\addtocounter{theorem}{0}

\begin{theorem}
$K^1_{USA}\left( R^3_{k,N} \right) = O\left( N^{\frac{1}{\left \lfloor \frac{k}{3} \right \rfloor}} + \log N \right)$.
\end{theorem}

\begin{figure}[htp]
	\begin{center}
    \begin{subfigure}[t]{.48\textwidth}
    \begin{center}
		\includegraphics[width=60px]{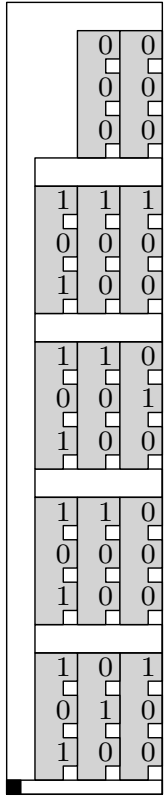}
        \caption{\label{fig:super_high_level_overview}  High-level: values of the counter. }
    \end{center}
    \end{subfigure}
    ~
    \begin{subfigure}[t]{.48\textwidth}
    \begin{center}
		\includegraphics[width=90px]{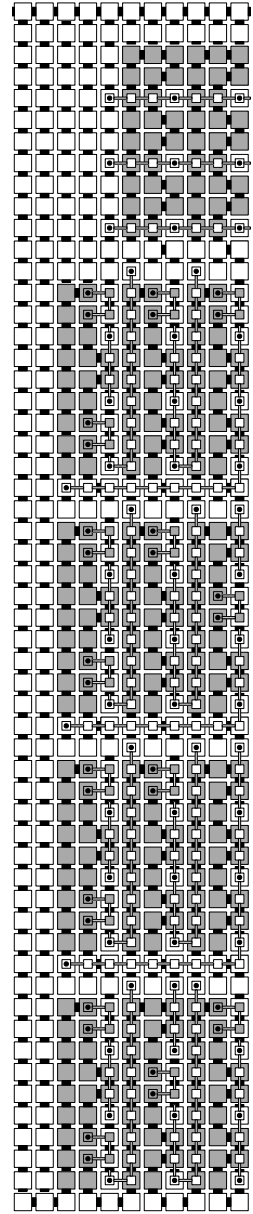}
        \caption{\label{fig:Example_Full} Low-level: full example. }
    \end{center}
    \end{subfigure}
    	 \caption{\label{fig:example_low_high} A full example of a $11\times 56$ construction.  The counter begins at 10-01-10 and is counting in ternary, so the initial value is $23 = 2\cdot 3^2 + 1\cdot 3^1 + 2 \cdot 3^0$. Note that the least significant bit for each digit, which is the lowest bit in each digit column, is actually a ``left edge'' indicator, where ``1'' means ``leftmost'' and ``0'' means ``not leftmost''.   To help distinguish overlapping tiles, ``write gadgets'' are drawn in gray.}
	\end{center}
\end{figure}

The basic idea of our construction for Theorem~\ref{thm:two} is to use a counter, the base of which is a function of $k$ and $N$. Then, we initialize the counter for our construction with a certain starting value and have it increment until its maximum value, at which point it rolls over to all 0's and the assembly goes terminal. Our construction is inspired by, but substantially different from, a similar construction by Aggarwal, Cheng, Goldwasser, Kao, Moisset de Espan\'{e}s and Schweller \cite{AGKS05g} for the self-assembly of two-dimensional $k \times N$ thin rectangles at temperature-2. Like theirs, our construction uses a counter whose base depends on $k$ and $N$. But unlike theirs, we represent each digit of the counter in our construction geometrically, using a one-bit-per-bump representation in an assembly that is three tiles wide and whose height is proportional to the binary representation of the base of the counter. See Figure~\ref{fig:example_low_high} for an example of the counter in our construction at different levels of granularity. The size of the tile set produced by the construction is $O\left(m + \log N\right) = O\left( N^{\frac{1}{\left \lfloor \frac{k}{3} \right \rfloor}} + \log N\right)$ and unique self-assembly follows from conditional determinism. The full construction of the tile set that proves  Theorem~\ref{thm:two} is given in the remainder of this section.

\subsection{Notation for gadgets and figures}

In the context of the construction of a tile set, a \emph{gadget}, referred to by a name like {\tt Gadget}, is a group of tiles that perform a specific task as they self-assemble. All gadgets are depicted in a corresponding figure, where the input glue is explicitly specified by an arrow, output glues are inferred and glues internal to the gadget are configured to ensure unique self-assembly within the gadget. We say that a gadget is \emph{general} if its input and output glues are undefined. If {\tt Gadget} is a general gadget, then we use the notation ${\tt Gadget}({\tt a}, {\tt b})$ to represent the creation of the \emph{specific gadget}, or simply \emph{gadget}, referred to as {\tt Gadget}, with input glue label {\tt a} and output glue label {\tt b} (all positive glue strengths are $1$). If a gadget has two possible output glues, then we will use the notation ${\tt Gadget}({\tt a}, {\tt b}, {\tt c})$ to denote the specific version of {\tt Gadget}, where {\tt a} is the input glue and {\tt b} and {\tt c} are the two possible output glues, listed in the order north, east, south and west, with all of the $z=0$ output glues listed before the $z=1$ output glues. If a gadget has only one output glue (and no input glue), like a gadget that contains the seed, or if a gadget has only one input glue (and no output glue), then we will use the notation ${\tt Gadget}({\tt a})$. We use the notation $\langle \cdot \rangle$ to denote some standard encoding of the concatenation of a list of symbols.

Following standard presentation conventions for ``just-barely'' 3D tile self-assembly, we use big squares to represent tiles placed in the $z=0$ plane and small squares to represent tiles placed in the $z=1$ plane. A glue between a $z=0$ tile and a $z=1$ tile is denoted as a small black disk. Glues between $z=0$ tiles are denoted as thick lines. Glues between $z=1$ tiles are denoted as thin lines. The following is our main positive result.

\subsection{Parameters for the counter}

Since the height (number of tile rows) of each logical row in the counter depends on $k$ and $N$, we must choose its starting value carefully. Therefore, let  $d=\left \lfloor \frac{k}{3} \right \rfloor$, $m=\left\lceil\left(\frac{N}{5}\right)^{\frac{1}{d}}\right\rceil$, $l=\left\lceil\log m\right\rceil+1$, $s=m^d-\left\lfloor\frac{N-3l-1}{3l+2}\right\rfloor$, $c=k\mod3$, and $r=N+1\mod 3l+2$, where $d$ is the number of digits in the counter, $m$ is the numerical base of the counter, $l$ is the number of bits needed to represent each digit in the counter's value plus one for the ``left edge'', $s$ is the numerical start of the counter, and $c$ and $r$ are the number of tile columns and tile rows, respectively, that must be filled in after and outside of the counter.  Each digit of the counter requires a width of 3 tiles, which has a direct relation with the tile complexity of the construction.  The values of $m$ and $s$ are chosen such that the counter stops just before reaching a height of $N$ tiles, at which point, the assembly is given a flat ``roof''. For example, in Figure~\ref{fig:example_low_high}, we have $d=3$, $m=3$, $l=3$, $s=23$, $c=2$, and $r=2$. We now informally justify each of the previously-defined parameters.

We define $d$ as the number of digit columns in the counter.  The number of digits in the counter is limited by the width, $k$.  Each digit column requires $3$ tiles.  We maximize the range of the counter by maximizing the number of digit columns. Given those requirements and restrictions, we end up with $d=\left\lfloor\frac{k}{3}\right\rfloor$ digit columns.

Given at least one digit column, we have the ability to count to any number of counter rows so long as we choose an appropriate base.  In that sense, choosing $m$, the base of the counter, is somewhat arbitrary.  So long as the maximum height of our construction, where we count $m^d$ times, is greater than $N$, and the minimum height of our construction, where we count once, is less than $N$, then there exists a corresponding value of $s$ such that the counter stops just before reaching a height of $N$.  We choose $m=\left\lceil\left(\frac{N}{5}\right)^{\frac{1}{d}}\right\rceil$.

The value of $l$, which is the number of binary bits needed to encode each digit of base-$m$, plus one for the ``left edge'', is easily verified by inspection.

Let $h$ be the height of the construction without any additional ``roof'' tiles.  By inspection of the example construction in Figure~\ref{fig:Example_Full}, as well as the construction of the gadget units in subsequent sections, we see that the height of the construction without additional tiles is $3l+1$ for the \texttt{Seed} unit (see Section~\ref{sec:seed_unit}), plus $3l+2$ for each incrementation of the counter, adding a Counter unit (see Section~\ref{sec:counter_unit}) row each time.  If we define $n$ as the number of Counter unit rows, then $h=n(3l+2)+3l+1$.  So then the maximum height of the counter is $m^d(3l+2)+3l+1$.

\begin{lemma}
\label{lem:max-height-counter}
$N \leq m^d(3l+2)+3l+1$.
\end{lemma}

\begin{proof}
We have
\begin{eqnarray*}
  N 	& 	= 		& 	5\left(\frac{N}{5}\right) \
     	   	=		 \	5\left( \left(\frac{N}{5}\right)^{\frac{1}{d}} \right)^d \
	 	\leq		 \	5 \left \lceil \left( \frac{N}{5} \right)^{\frac{1}{d}} \right \rceil^d \\
	&	=	 	& 	5m^d \
	 	\leq	 	\	3lm^d+2m^d \
	 	\leq 		\	3lm^d+2m^d+3l+1 \\
	&	=		&	 m^d(3l+2)+3l+1.
\end{eqnarray*}
\end{proof}

And the minimum height is $6l+3$.

\begin{lemma}
\label{lem:min-height-counter}
$6l+3 \leq N$, for $k\geq 3$ and sufficiently large values of $N$.
\end{lemma}

\begin{proof}
We have
\begin{eqnarray*}
6l+3 & = 		& 6 \left \lceil \log m \right \rceil + 9 \
	\ =	\	 6 \left \lceil \log \left \lceil \left( \frac{N}{5} \right)^{\frac{1}{d}} \right \rceil \right \rceil + 9 \\
	& =		& 6 \left \lceil \log \left( \left( \frac{N}{5} \right)^{\frac{1}{d}} \right) \right \rceil + 9 \quad\quad\quad\quad\quad \textmd{See \cite{Graham:1994} for a proof of this equality} \\
	& \leq 	& 6 \log \left( \left( \frac{N}{5} \right)^{\frac{1}{d}} \right) + 15
	\ = 		\ \frac{6}{d} \log \left( \frac{N}{5}  \right) + 15 \\
	& \leq	& \frac{6}{d} \log N  + 15 
	\ =		\ \frac{6}{\left \lfloor \frac{k}{3} \right \rfloor } \log N+ 15 \\
	& \leq	& 6 \log N+ 15\quad\quad\quad\quad\quad\quad\textmd{Because } k \geq 3 \\
	& \leq	& N \quad\quad\quad\quad\quad\quad\quad\quad\quad\quad \textmd{For } N \geq 49.
\end{eqnarray*}
\end{proof}

By Lemma~\ref{lem:min-height-counter}, one row of the counter might result in a final assembly that will not be tall enough but Lemma~\ref{lem:max-height-counter} says that having all possible rows of the counter might result in a final assembly that is too tall. Therefore, we must start the counter at an appropriate value to get the correct height of the final assembly. 

The counter can start at any whole number less than $m^d$ and ends when it reaches $0$ by rolling over $m^d-1$.  This means that the number of Counter unit rows $n$, is $m^d-s$, where we have defined $s$ as the starting value of the counter.  To choose the best starting value, we find the value for $n$ that gets $h$ close to $N$ without exceeding $N$.  It follows from the equation $h=n(3l+2)+3l+1$, that $n=\left\lfloor\frac{N-3l-1}{3l+2}\right\rfloor$.  Thus, $s=m^d-\left\lfloor\frac{N-3l-1}{3l+2}\right\rfloor$.

Recall that we must use $3d$ tile columns to encode the digits of the counter.  Since the remaining tile columns must still be filled in to ensure a width of $k$, we must sometimes construct additional tile columns outside of the counter.  We denote the number of additional tile columns by $c$ and conclude that its value is $k\mod3$.

Similarly, we must also account for the number of additional tile rows that must be filled in after the counter has finished counting.  We denote the number of filler rows at the top of the construction with $r$, and conclude that it is the remainder of this quotient expression $\frac{N-3l-1}{3l+2}$, which is to say that $r=N-3l-1\mod 3l+2$.  This modular expression simplifies to $r=N+1\mod 3l+2$.

\subsection{Tile set}

For the purposes of this construction, we assume the function $bin(a,b)$ gives the binary representation of $a$ that is truncated (or prepended with 0's) to a length of $b$.  We will use the notation $\left(a\right)_b\left[i\right]$ to denote the digit in the $i$-th position from the right of $a$ in base-$b$.  We define a \emph{gadget unit} as a collection of gadgets with a singular purpose.  Gadgets belonging to the same gadget unit will have their figures grouped together. The set of all gadgets created in this subsection forms the tile set $T_{k,N}$. 

\subsubsection{{\tt Vertical\_Column} tiles}

Since  \texttt{Vertical\_Column} tiles are present in a majority of our gadget units, they will only be shown in Figure~\ref{fig:Vertical_Column}.

\begin{figure}[htp]
	\begin{center}
    \begin{subfigure}[t]{.4\textwidth}
    \begin{center}
		\includegraphics[width=12px]{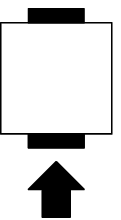}
        \caption{\label{fig:Up_Column_Tile}\texttt{Up\_Column\_Tile}}
    \end{center}
    \end{subfigure}
    \begin{subfigure}[t]{.4\textwidth}
    \begin{center}
		\includegraphics[width=12px]{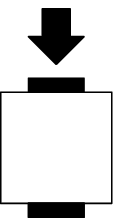}
        \caption{\label{fig:Down_Column_Tile}\texttt{Down\_Column\_Tile}}
    \end{center}
    \end{subfigure}
    	\caption{\label{fig:Vertical_Column}\texttt{Vertical\_Column} tiles are used throughout the construction to adjust the height of gadget units.}
	\end{center}
\end{figure}

\subsubsection{Seed unit}

\label{sec:seed_unit}
We begin by encoding the initial value of the counter with the Seed unit.  It has $d$ columns, where each 3-wide column represents a digit of $s$ in base-$m$.  A collection of bit-bumps on the columns' east sides encodes the digits into binary.  A small ``lip'' is added on the west side of the Seed unit to increase the width of the assembly by $c$, which catches any vertical filler tiles at the end of the construction.  The \texttt{Guess} tile on the east side of the unit initiates the first Counter unit.  See Figure~\ref{fig:Seed}.
    
We define the Seed unit by creating the following gadgets:

\begin{itemize}

\item The first gadget depends on the value of $c$:

\begin{itemize}

\item If $c=0$, create $\texttt{Seed\_Start}\left(\left<\texttt{seed},\texttt{col},d,1\right>\right)$ from the general gadget in Figure~\ref{fig:Seed_Start_0}.

\item If $c=1$, create $\texttt{Seed\_Start}\left(\left<\texttt{seed},\texttt{col},d,1\right>\right)$ from the general gadget in Figure~\ref{fig:Seed_Start_1}.

\item If $c=2$, create $\texttt{Seed\_Start}\left(\left<\texttt{seed},\texttt{col},d,1\right>\right)$ from the general gadget in Figure~\ref{fig:Seed_Start_2}.

\end{itemize}

One gadget was created in this step.

\item For each $i=1,\ldots,d$:

\begin{itemize}

\item For each $j=1,\ldots,3l-3$, create\\$\texttt{Up\_Column}\left(\left<\texttt{seed},\texttt{col},i,j\right> \! ,\left<\texttt{seed},\texttt{col},i,j+1\right>\right)$ from the general gadget in Figure~\ref{fig:Up_Column_Tile}.

\item Create $\texttt{Seed\_Msb}\left(\left<\texttt{seed},\texttt{col},i,3l-2\right> \! ,\left<\texttt{seed},\texttt{bit},i,l-1\right>\right)$ from the general gadget in Figure~\ref{fig:Seed_Msb_0} if $((s)_m[i])_2[l]=0$ or Figure~\ref{fig:Seed_Msb_1} if $((s)_m[i])_2[l]=1$.

\item For each $j=2,\ldots,l-1$, create\\$\texttt{Seed\_Bit}\left(\left<\texttt{seed},\texttt{bit},i,j\right> \! ,\left<\texttt{seed},\texttt{bit},i,j-1\right>\right)$ from the general gadget in Figure~\ref{fig:Seed_Bit_0} if $((s)_m[i])_2[j]=0$ or Figure~\ref{fig:Seed_Bit_1} if $((s)_m[i])_2[j]=1$.

\end{itemize}

In this step, $\sum_{i=1}^d{\left(1 + \sum_{j=1}^{3l-3}{1} + \sum_{j=2}^{l-1}{1}\right)} = O(dl)$ gadgets were created. This means that
\begin{eqnarray*}
dl  				& = & O\left(d \log m\right) \\
				& = & O\left( d \log \left \lceil \left(\frac{N}{5} \right)^{\frac{1}{d}} \right \rceil \right) \\
				& = & O\left( d \log \left( 2\left( \frac{N}{5} \right)^{\frac{1}{d}} \right) \right) \\
				& = & O\left( d \log \left( \frac{N}{5} \right)^{\frac{1}{d}} + d \log 2\right) \\
				& = & O\left( \log N + \left \lfloor \frac{k}{3} \right \rfloor \right) \\
				& = & O\left( \log N \right)
\end{eqnarray*}
gadgets were created in this step. 

\item Create $\texttt{Seed\_Bit}\left(\left<\texttt{seed},\texttt{bit},d,1\right> \! ,\left<\texttt{seed},\texttt{bit},d,0\right>\right)$ from the general gadget in Figure~\ref{fig:Seed_Bit_1}. One gadget was created in this step. 

\item For each $i=1,\ldots,d-1$:

\begin{itemize}

\item Create $\texttt{Seed\_Bit}\left(\left<\texttt{seed},\texttt{bit},i,1\right> \! ,\left<\texttt{seed},\texttt{bit},i,0\right>\right)$ from the general gadget in Figure~\ref{fig:Seed_Bit_0}.

\item Create $\texttt{Seed\_Spacer}\left(\left<\texttt{seed},\texttt{bit},i+1,0\right> \! ,\left<\texttt{seed},\texttt{col},i,1\right>\right)$ from the general gadget in Figure~\ref{fig:Seed_Spacer}.

\end{itemize}

In this step, $2(d-1) = O(k) = O(\log N)$ gadgets were created. 

\item Create 
	\begin{align*}
		\texttt{Seed\_End}( 		& \! \left<\texttt{seed},\texttt{bit},1,0\right> \!, \\
							& \! \left<\texttt{inc},\texttt{read},1\right> \!, \\
							& \! \left<\texttt{inc},\texttt{read},0\right> )
	\end{align*}
One gadget was created in this step.

\end{itemize}
    
\begin{figure}[htp]
	\begin{center}
    \begin{subfigure}[b]{.2\textwidth}
    \begin{center}
		\includegraphics[width=12px]{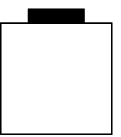}
        \caption{\label{fig:Seed_Start_0}\texttt{Seed\_Start\_0}}
        \vspace{.56\linewidth}
        \includegraphics[width=27px]{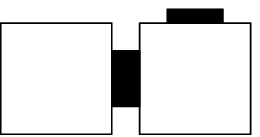}
        \caption{\label{fig:Seed_Start_1}\texttt{Seed\_Start\_1}}
        \vspace{.3\linewidth}
        \includegraphics[width=42px]{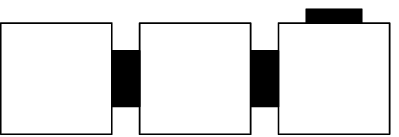}
        \caption{\label{fig:Seed_Start_2}\texttt{Seed\_Start\_2}}
    \end{center}
    \end{subfigure}
    \begin{subfigure}[b]{.2\textwidth}
    \begin{center}
		\includegraphics[width=42px]{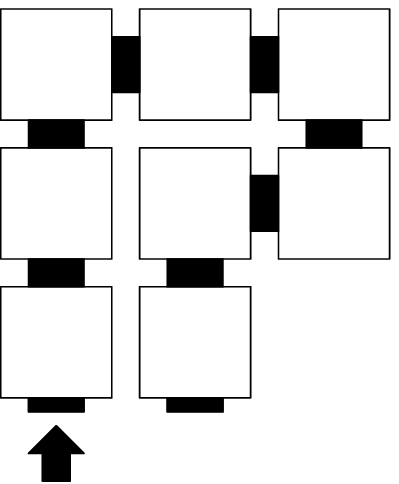}
        \caption{\label{fig:Seed_Msb_0}\texttt{Seed\_Msb\_0}}
        \vspace{.17\linewidth}
        \includegraphics[width=39px]{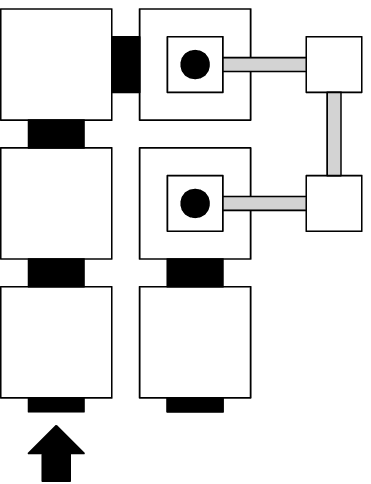}
        \caption{\label{fig:Seed_Msb_1}\texttt{Seed\_Msb\_1}}
        \vspace{.22\linewidth}
        \includegraphics[width=42px]{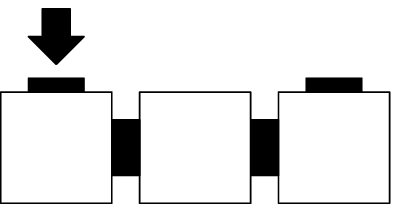}
        \caption{\label{fig:Seed_Spacer}\texttt{Seed\_Spacer}}
    \end{center}
    \end{subfigure}
    \begin{subfigure}[b]{.2\textwidth}
    \begin{center}
		\includegraphics[width=27px]{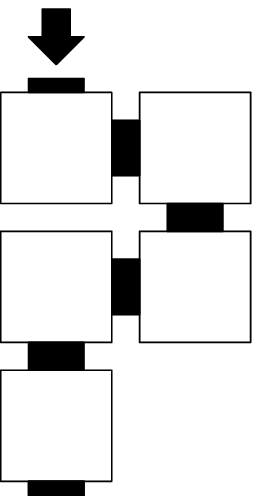}
        \caption{\label{fig:Seed_Bit_0}\texttt{Seed\_Bit\_0}}
        \vspace{.15\linewidth}
        \includegraphics[width=24px]{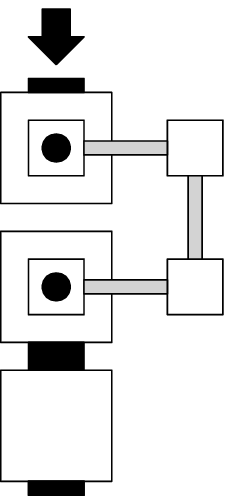}
        \caption{\label{fig:Seed_Bit_1}\texttt{Seed\_Bit\_1}}
        \vspace{.15\linewidth}
        \includegraphics[width=27px]{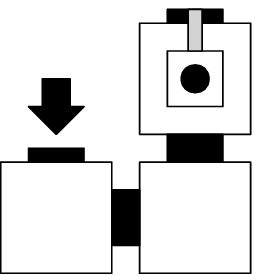}
        \caption{\label{fig:Seed_End}\texttt{Seed\_End}}
    \end{center}
    \end{subfigure}
    \vspace{20pt}

    \begin{subfigure}[b]{.38\textwidth}
    \begin{center}
		\includegraphics[width=162px]{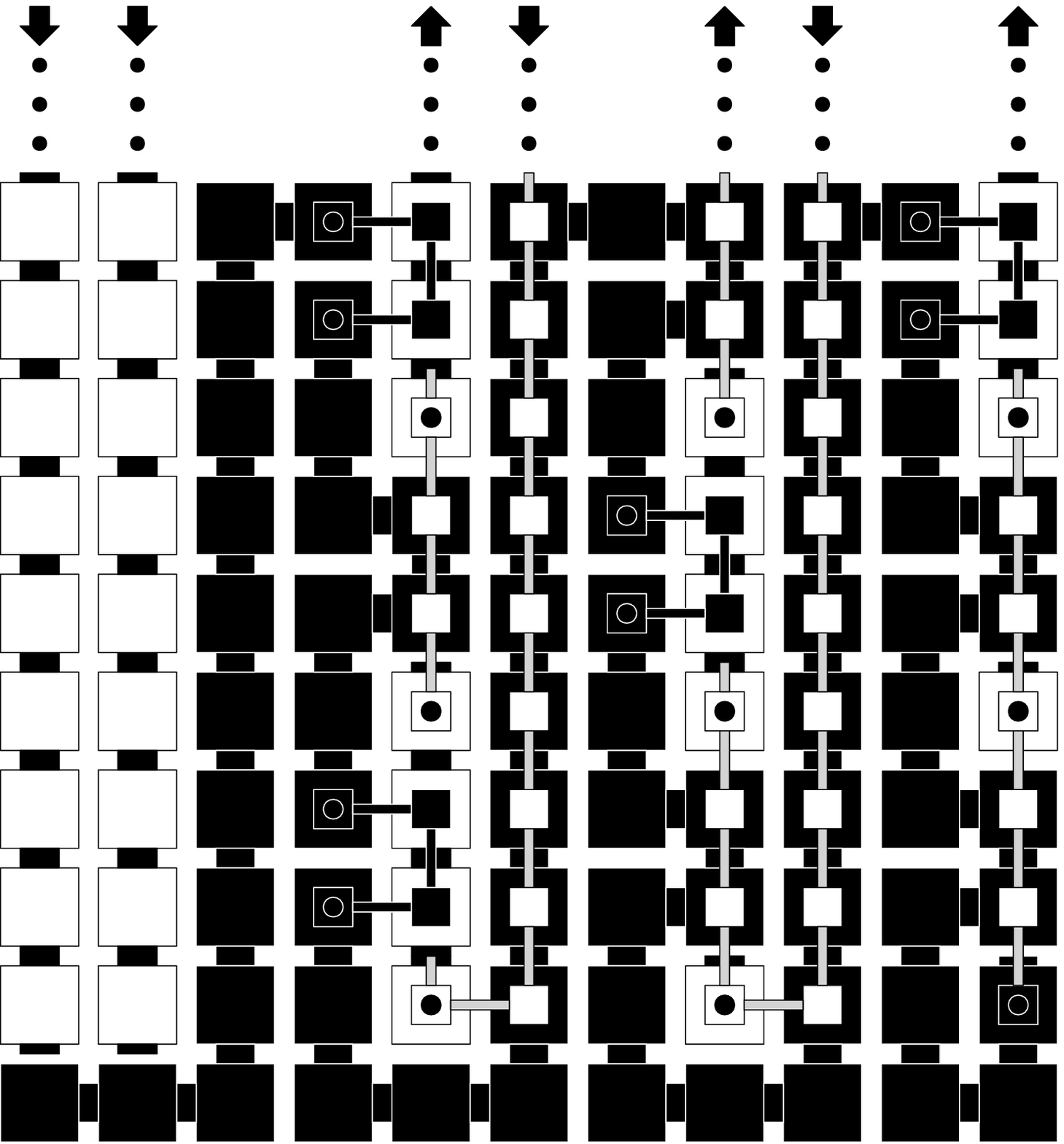}
        \caption{\label{fig:Example_Seed} The seed is the leftmost tile in the bottom row and self-assembly of the black tiles proceeds in a left-to-right fashion.}
    \end{center}
    \end{subfigure}

    	\caption{\label{fig:Seed}The Seed gadget unit.  The actual seed tile is at the far-left of any of the \texttt{Seed\_Start} gadgets.}
	\end{center}
\end{figure}

Since the number of tiles in each gadget in the Seed unit is $O(1)$, then based on the above computations, the number of gadgets created in this subsection is $O(\log N)$.

\subsubsection{Counter unit}
\label{sec:counter_unit}

For this construction, we require a set of $4m-1$ Counter gadget units to encode the digits of the counter, of which $m$ units will increment a digit of the counter, $m$ units will copy a digit of the counter, $m$ units will copy the most significant digit of the counter, and $m-1$ units will increment the most significant digit of the counter.  Each row of the counter will have a total of $d$ Counter units.  Each Counter unit reads over a series of bit-bumps protruding into their row from the preceding Seed unit or counter row.  After a Counter unit reads its bit pattern with \texttt{Guess} tiles, the unit produces a new bit pattern in the row above it that encodes a copy or increment of the current digit.  Of the ``less significant digit'' units, the one that increments $m-1$ to $0$ is unique because it initiates another increment unit, that is, a carry is passed to the digit to its left.  Other increment units, as well as the copy units, will only initiate copy units (no carry is propagated to the left).  The first bit read is always the ``left edge'' marker, which tells the unit if it represents the most significant digit of the counter value and needs to start a new row instead of another Counter unit.  The counter terminates when the most significant digit follows an increment unit and reads $m-1$ in its column.  At that point, the counter will have rolled over $m^d-1$ and the Roof unit takes over.  The gadgets belonging to the Counter units are shown in Figure~\ref{fig:Counter}.

We define the Counter units by creating the following gadgets:

\begin{itemize}

\item For each $i = 0, \ldots, l - 2$ and each $u \in \{0,1\}^i$: 

\begin{itemize}

\item Create
	\begin{align*}
		\texttt{Counter\_Read}( 			& \! \left<\texttt{inc},\texttt{read},0u\right> \!, \\
									& \! \left<\texttt{inc},\texttt{read},10u\right> \!, \\
									& \! \left<\texttt{inc},\texttt{read},00u\right> )
	\end{align*}
from the general gadget in Figure~\ref{fig:Counter_Read_0}.

\item Create
	\begin{align*}
		\texttt{Counter\_Read}( 			& \! \left<\texttt{inc},\texttt{read},1u\right> \!, \\
									& \! \left<\texttt{inc},\texttt{read},11u\right> \!, \\
									& \! \left<\texttt{inc},\texttt{read},01u\right> )
	\end{align*}
from the general gadget in Figure~\ref{fig:Counter_Read_1}.

\item Create
	\begin{align*}
		\texttt{Counter\_Read}( 		& \! \left<\texttt{copy},\texttt{read},0u\right> \!, \\
								& \! \left<\texttt{copy},\texttt{read},10u\right> \!, \\
								& \! \left<\texttt{copy},\texttt{read},00u\right> )
	\end{align*}
from the general gadget in Figure~\ref{fig:Counter_Read_0}.

\item Create
	\begin{align*}
		\texttt{Counter\_Read}( 		& \! \left<\texttt{copy},\texttt{read},1u\right> \!, \\
								& \! \left<\texttt{copy},\texttt{read},11u\right> \!, \\
								& \! \left<\texttt{copy},\texttt{read},01u\right> )
	\end{align*}
from the general gadget in Figure~\ref{fig:Counter_Read_1}.

\end{itemize}

These are read gadgets for all digit positions and all bits (except the most significant bit) for both copy and increment columns. In this step, $\sum_{i=0}^{l-2}{4\cdot 2^i} = 4 \left( 2^{l-1} - 1 \right) = 4 \left( 2^{\left \lceil \log m \right \rceil } - 1\right) \leq 4 \left( 2\cdot 2^{\log m}\right) = O(m) = O\left(N^{\frac{1}{\left \lfloor \frac{k}{3} \right \rfloor}}\right)$ gadgets were created. 

\end{itemize}

Recall that the least significant bit of our digit columns represents whether the digit is for the most significant digit's place or not.  It follows then, that $bin(2m-2,l)$ is the greatest digit of our base-$m$ counter but with an extra $0$ at the end, which is to say that it is m-1 but not for the most significant digit's place.  To get $m-1$ for the most significant digit's place, we simply add a $1$ to that expression and get $bin(2m-1,l)$.  Hence, we use the range $0$ to $2m-1$ to refer to all digits of base-$m$ in all positions / places, we use the range $0$ to $2m-3$ to refer to all digits of base-$m$ except $m-1$ in all positions / places, we use the index $2m-2$ to refer to the counter value $m-1$ when not in the most significant digit's place, and we use the index $2m-1$ to refer to the counter value $m-1$ when in the most significant digit's place.

\begin{itemize}

\item For each $i=0,\ldots,2m-1$, create
	\begin{align*}
			{\tt Counter\_Read\_Msb}( & \! \left \langle \texttt{copy},\texttt{read},bin(i,l) \right \rangle \! , \\
								&  \! \left \langle  \texttt{copy},\texttt{write},bin(i,l) \right \rangle \! , \\
								&  \! \left \langle \texttt{d\_fill} \right \rangle)
	\end{align*}
	from the general gadget in Figure~\ref{fig:Counter_Read_Msb_0} if $(i)_2[l]=0$ or Figure~\ref{fig:Counter_Read_Msb_1} if $(i)_2[l]=1$. These are all read gadgets for the most significant bit in copy columns only. In this step, $2m = O\left(N^{\frac{1}{\left \lfloor \frac{k}{3} \right \rfloor}}\right)$ gadgets were created.

\item For each $i=0,\ldots,2m-3$, create
	\begin{align*}
			{\tt Counter\_Read\_Msb}( & \! \left \langle  \texttt{inc},\texttt{read},bin(i,l) \right \rangle \! , \\
								&  \! \left \langle  \texttt{copy},\texttt{write},bin(i+2,l) \right \rangle \! , \\
								&  \! \left \langle \texttt{d\_fill} \right \rangle )
	\end{align*}
	from the general gadget in Figure~\ref{fig:Counter_Read_Msb_0} if $(i)_2[l]=0$ or Figure~\ref{fig:Counter_Read_Msb_1} if $(i)_2[l]=1$. These are all read gadgets for the most significant bit in increment columns but only when the digit value is between $0$ and $m - 2$. In this step, $2m-2 = O\left(N^{\frac{1}{\left \lfloor \frac{k}{3} \right \rfloor}}\right)$ gadgets were created.

\item Create 
	\begin{align*}
		\texttt{Counter\_Read\_Msb}( 	& \! \left \langle \texttt{inc},\texttt{read},bin(2m-2,l)\right \rangle \!, \\
								& \! \left \langle \texttt{inc},\texttt{write\_all\_0s},1 \right \rangle \!, \\
								& \! \left \langle \texttt{d\_fill} \right \rangle )
	\end{align*}
from the general gadget in Figure~\ref{fig:Counter_Read_Msb_1}. This is the read gadget for the most significant bit in all increment columns (except the most significant digit) but only when the digit value is $m - 1$. One gadget was created in this step.

\item For each $i = 1, \ldots, l - 1$, create\\${\tt Counter\_Write}(\langle {\tt inc}, {\tt write\_all\_0s}, i\rangle, \langle {\tt inc}, {\tt write\_all\_0s}, i+1\rangle)$ from the general gadget in Figure~\ref{fig:Counter_Write_0}. These are the write gadgets for a digit value of all 0s due to incrementation of $m -1$. In this step, $l-1 = O(\log m) = O\left(\frac{\log N}{k}\right) = O(\log N)$ gadgets were created. This group of gadgets writes a series of 0 bits, because $m-1$ was incremented to $0$.

\item For each $u \in \{0,1\}^{l-1}$:

\begin{itemize}

\item Create $\texttt{Counter\_Write}\left(\left<\texttt{copy},\texttt{write},u0\right> \!,\left<\texttt{copy},\texttt{write},u\right>\right)$ from the general gadget in Figure~\ref{fig:Counter_Write_0}.

\item Create $\texttt{Counter\_Write}\left(\left<\texttt{copy},\texttt{write},u1\right> \! ,\left<\texttt{msd},\texttt{write},u\right>\right)$ from the general gadget in Figure~\ref{fig:Counter_Write_1}. Note that ``{\tt msd}'' stands for ``most significant digit'' (since the ``left edge'' bit is 1 in this case).

\end{itemize}

These are all gadgets to write the ``left edge'' marker in all copy columns. In this step, $2\cdot 2^{l-1} = 2\cdot 2^{\left \lceil \log m \right \rceil} \leq 2 \left( 2\cdot 2^{ \log m}\right) = O(m) = O\left(N^{\frac{1}{\left \lfloor \frac{k}{3} \right \rfloor}}\right)$ gadgets were created.

\item For each $i = 1, \ldots, l - 2$ and each $u \in \{0,1\}^i$:

\begin{itemize}

\item Create $\texttt{Counter\_Write}\left(\left<\texttt{copy},\texttt{write},u0\right> \!,\left<\texttt{copy},\texttt{write},u\right>\right)$ from the general gadget in Figure~\ref{fig:Counter_Write_0}.

\item Create $\texttt{Counter\_Write}\left(\left<\texttt{copy},\texttt{write},u1\right> \!,\left<\texttt{copy},\texttt{write},u\right>\right)$ from the general gadget in Figure~\ref{fig:Counter_Write_1}.

\item Create $\texttt{Counter\_Write}\left(\left<\texttt{msd},\texttt{write},u0\right> \!,\left<\texttt{msd},\texttt{write},u\right>\right)$ from the general gadget in Figure~\ref{fig:Counter_Write_0}.

\item Create $\texttt{Counter\_Write}\left(\left<\texttt{msd},\texttt{write},u1\right> \!,\left<\texttt{msd},\texttt{write},u\right>\right)$ from the general gadget in Figure~\ref{fig:Counter_Write_1}.

\end{itemize}

These gadgets write all digits (except the ``left edge'' marker) in all copy columns. In this step, 
\begin{eqnarray*}
\sum_{i=1}^{l-2}{4\cdot 2^i} & = & 4 \left( 2^{l-1} - 2 \right) \\
					  & = & 4 \left( 2^{\left \lceil \log m \right \rceil } - 2\right) \\
					  & \leq & 4 \left( 2\cdot 2^{\log m}\right) \\
					  & = & O(m) = O\left(N^{\frac{1}{\left \lfloor \frac{k}{3} \right \rfloor}}\right)
\end{eqnarray*}
gadgets were created. 

\item Create $\texttt{Counter\_Write\_Msb}\left(\left<\texttt{inc},\texttt{write},l\right> \!,\left<\texttt{inc},\texttt{down\_z\_0},1\right>\right)$ from the general gadget in Figure~\ref{fig:Counter_Write_Msb_0}. This gadget writes 0 for the most significant bit when digit value is all 0s after incrementing $m - 1$. One gadget was created in this step.

\item Create $\texttt{Counter\_Write\_Msb}\left(\left<\texttt{copy},\texttt{write},0\right> \!,\left<\texttt{copy},\texttt{down\_z\_0},1\right>\right)$ from the general gadget in Figure~\ref{fig:Counter_Write_Msb_0}. This gadget writes 0 for the most significant bit in any copy column except the most significant digit column. One gadget was created in this step.

\item Create $\texttt{Counter\_Write\_Msb}\left(\left<\texttt{copy},\texttt{write},1\right> \!,\left<\texttt{copy},\texttt{down\_z\_0},1\right>\right)$ from the general gadget in Figure~\ref{fig:Counter_Write_Msb_1}. This gadget writes 1 for the most significant bit in any copy column except the most significant digit column. One gadget was created in this step.

\item Create $\texttt{Counter\_Write\_Msb}\left(\left<\texttt{msd},\texttt{write},0\right> \!,\left<\texttt{msd},\texttt{down\_z\_0},1\right>\right)$ from the general gadget in Figure~\ref{fig:Counter_Write_Msb_0}. This gadget writes 0 for the most significant bit in the most significant digit (i.e., copy) column. One gadget was created in this step.

\item Create $\texttt{Counter\_Write\_Msb}\left(\left<\texttt{msd},\texttt{write},1\right> \!,\left<\texttt{msd},\texttt{down\_z\_0},1\right>\right)$ from the general gadget in Figure~\ref{fig:Counter_Write_Msb_1}. This gadget writes 1 for the most significant bit in the most significant digit (i.e., copy) column. One gadget was created in this step.

\item For each $i=1,\ldots,3l-2$:

\begin{itemize}

\item Create $\texttt{Down\_Column}\left(\left<\texttt{inc},\texttt{down\_z\_0},i\right> \!,\left<\texttt{inc},\texttt{down\_z\_0},i+1\right>\right)$ from the general gadget in Figure~\ref{fig:Down_Column_Tile}.

\item Create $\texttt{Down\_Column}\left(\left<\texttt{copy},\texttt{down\_z\_0},i\right> \!,\left<\texttt{copy},\texttt{down\_z\_0},i+1\right>\right)$ from the general gadget in Figure~\ref{fig:Down_Column_Tile}.

\end{itemize}

These gadgets go back down to the bottom of the new counter row in all columns (except the most significant digit column) while remembering an increment or copy is taking place. In this step, $2(3l-2) = O(l) = O(\log m) = O\left( \frac{\log N}{k} \right) = O(\log N)$ gadgets were created.

\item For each $i=1,\ldots,3l-3$, create\\$\texttt{Down\_Column}\left(\left<\texttt{msd},\texttt{down\_z\_0},i\right> \!,\left<\texttt{msd},\texttt{down\_z\_0},i+1\right>\right)$ from the general gadget in Figure~\ref{fig:Down_Column_Tile}. These gadgets go back down to the bottom of the new counter row in the most significant digit column while remembering an increment or copy is taking place. In this step, $3l-3 = O(l)=O(\log m) = O\left(\frac{\log N}{k}\right) = O(\log N)$ gadgets were created.

\item Create\\$\texttt{Counter\_Return\_Column\_Start}\left(\left<\texttt{inc},\texttt{down\_z\_0},3l-1\right> \!,\left<\texttt{inc},\texttt{down\_z\_1},1\right>\right)$ from the general gadget in Figure~\ref{fig:Counter_Return_Column_Start}. This is the first gadget to start going down to the bottom of the previous counter row in increment columns only. One gadget was created in this step.

\item Create\\$\texttt{Counter\_Return\_Column\_Start}\left(\left<\texttt{copy},\texttt{down\_z\_0},3l-1\right> \!,\left<\texttt{copy},\texttt{down\_z\_1},1\right>\right)$ from the general gadget in Figure~\ref{fig:Counter_Return_Column_Start}. This is the first gadget to start going down to the bottom of the previous counter row in copy columns only. One gadget was created in this step.


\item For each $i=1,\ldots,l-1$:

\begin{itemize}

\item Create\\$\texttt{Counter\_Return\_Column}\left(\left<\texttt{inc},\texttt{down\_z\_1},i\right> \!,\left<\texttt{inc},\texttt{down\_z\_1},i+1\right>\right)$ from the general gadget in Figure~\ref{fig:Counter_Return_Column}.

\item Create\\$\texttt{Counter\_Return\_Column}\left(\left<\texttt{copy},\texttt{down\_z\_1},i\right> \!,\left<\texttt{copy},\texttt{down\_z\_1},i+1\right>\right)$ from the general gadget in Figure~\ref{fig:Counter_Return_Column}.

\end{itemize}

These gadgets go back down to the bottom of the previous counter row in all columns (except the most significant digit column) while remembering an increment or copy is taking place. In this step, $2(l-1) = O(l) = O(\log m) = O\left( \frac{\log N}{k} \right) = O(\log N)$ gadgets were created.

\item Create
	\begin{align*}
		\texttt{Counter\_Return\_Column\_End}( 	& \! \left<\texttt{inc},\texttt{down\_z\_1},l\right> \!, \\
										& \! \left<\texttt{inc},\texttt{read},1\right> \!, \\
										& \! \left<\texttt{inc},\texttt{read},0\right> )
	\end{align*}
from the general gadget in Figure~\ref{fig:Counter_Return_Column_End}. This gadget gets ready to read the ``left edge'' marker in the next digit (to the left) in order to increment it. One gadget was created in this step.

\item Create
	\begin{align*}
		\texttt{Counter\_Return\_Column\_End}( 	& \! \left<\texttt{copy},\texttt{down\_z\_1},l\right> \!, \\
										& \! \left<\texttt{copy},\texttt{read},1\right> \!, \\
										& \! \left<\texttt{copy},\texttt{read},0\right> )
	\end{align*}
from the general gadget in Figure~\ref{fig:Counter_Return_Column_End}. This gadget gets ready to read the ``left edge'' marker in the next digit (to the left) in order to copy it. One gadget was created in this step.

\end{itemize}

Since the number of tiles in each gadget in the Counter unit is $O(1)$, then based on the above computations, the number of  gadgets (and therefore the number of tile types) created in this subsection is $O\left(N^{\frac{1}{\left \lfloor \frac{k}{3} \right \rfloor}} + \log N\right)$.
    
\begin{figure}[htp]
    \begin{subfigure}[b]{.3\textwidth}
    	\centering
	\includegraphics[width=12px]{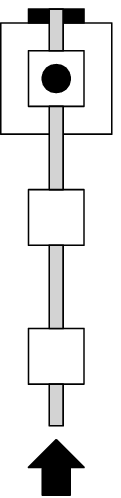}
        \caption{\label{fig:Counter_Read_0}\texttt{Counter\_Read\_0}}
     \end{subfigure}
     \begin{subfigure}[b]{.3\textwidth}
     	\centering
        \includegraphics[width=12px]{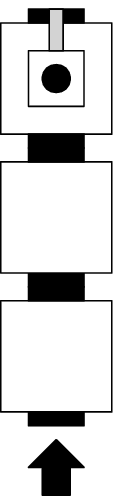}
        \caption{\label{fig:Counter_Read_1}\texttt{Counter\_Read\_1}}
     \end{subfigure}
     \begin{subfigure}[b]{.3\textwidth}
     	\centering
        \includegraphics[width=27px]{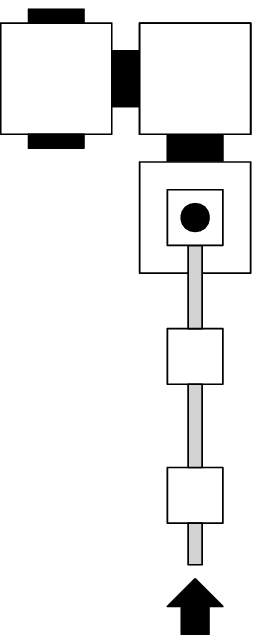}
        \caption{\label{fig:Counter_Read_Msb_0}\texttt{Counter\_Read\_Msb\_0}}
     \end{subfigure}
     \vspace{10pt}

     \begin{subfigure}[b]{.3\textwidth}
     	\centering
        \includegraphics[width=27px]{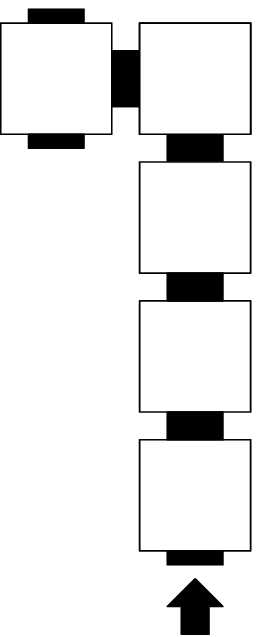}
        \caption{\label{fig:Counter_Read_Msb_1}\texttt{Counter\_Read\_Msb\_1}}
    \end{subfigure}
    \begin{subfigure}[b]{.3\textwidth}
    	\centering
	\includegraphics[width=27px]{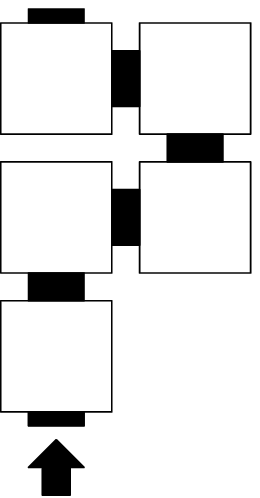}
        \caption{\label{fig:Counter_Write_0}\texttt{Counter\_Write\_0}}
     \end{subfigure}
     \begin{subfigure}[b]{.3\textwidth}
     	\centering
        \includegraphics[width=24px]{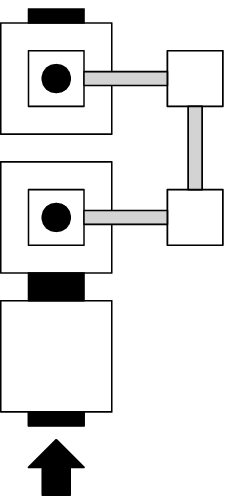}
        \caption{\label{fig:Counter_Write_1}\texttt{Counter\_Write\_1}}
    \end{subfigure}
    \vspace{10pt}

     \begin{subfigure}[b]{.3\textwidth}
     	\centering
        \includegraphics[width=42px]{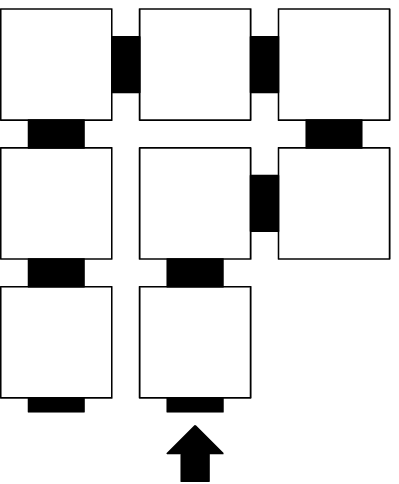}
        \caption{\label{fig:Counter_Write_Msb_0}\texttt{Counter\_Write\_Msb\_0}}
    \end{subfigure}
     \begin{subfigure}[b]{.3\textwidth}
     	\centering
        \includegraphics[width=39px]{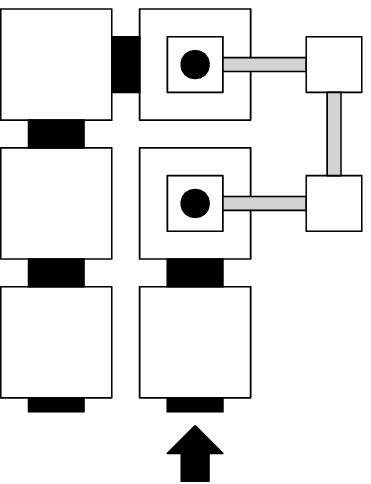}
        \caption{\label{fig:Counter_Write_Msb_1}\texttt{Counter\_Write\_Msb\_1}}
    \end{subfigure}
    \begin{subfigure}[b]{.3\textwidth}
    	\centering
	\includegraphics[width=12px]{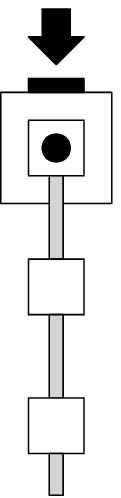}
        \caption{\label{fig:Counter_Return_Column_Start}\texttt{Counter\_Return\_Column\_Start}}
     \end{subfigure}
     \vspace{10pt}

     \begin{subfigure}[b]{.3\textwidth}
     	\centering
        \includegraphics[width=6px]{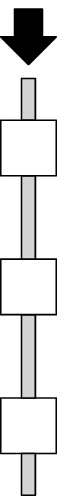}
        \caption{\label{fig:Counter_Return_Column}\texttt{Counter\_Return\_Column}}
     \end{subfigure}
     \begin{subfigure}[b]{.3\textwidth}
     	\centering
        \includegraphics[width=24px]{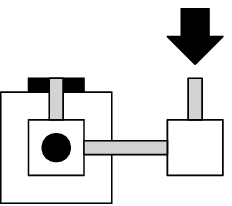}
        \caption{\label{fig:Counter_Return_Column_End}\texttt{Counter\_Return\_Column\_End}}
    \end{subfigure}	
       	\caption{\label{fig:Counter}The Counter gadget unit.  A row of units shares half of its space with the row above and the other half of its space with the row below.}
\end{figure}

\begin{figure}[htp]
	\centering
	\includegraphics[width=120px]{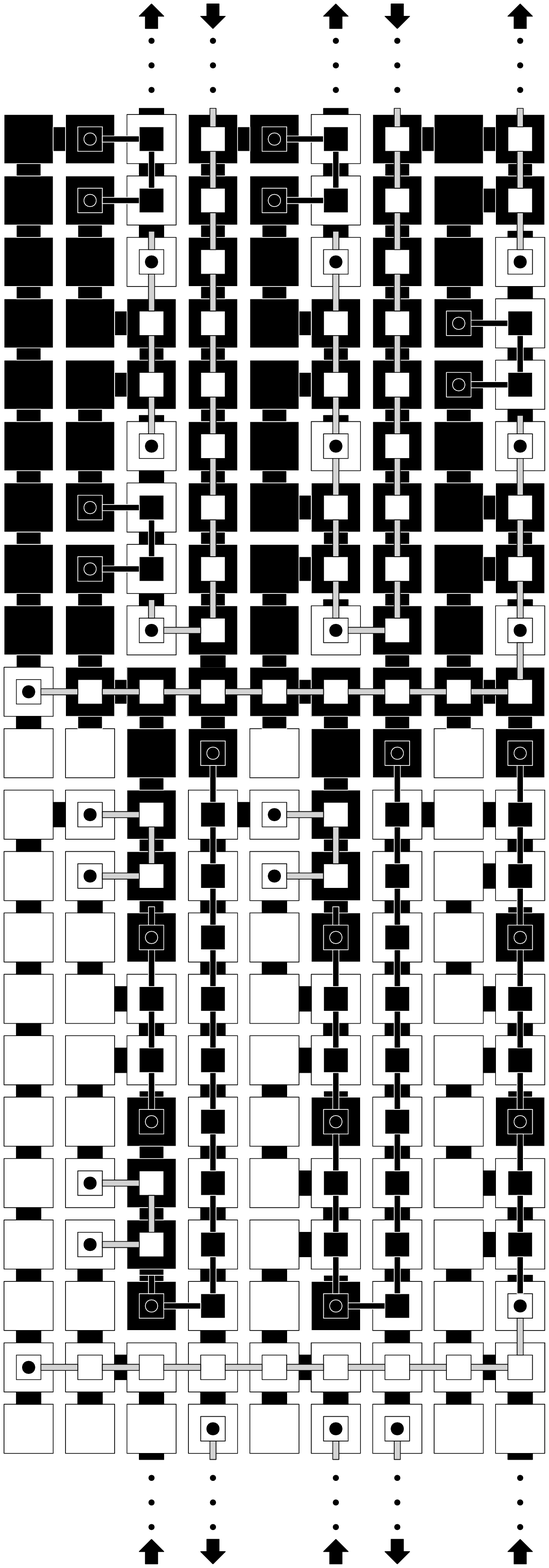}
        \caption{\label{fig:Example_Counter} An example with a row of Counter units colored black.  The unit is incrementing 10-10-00 to 10-10-01.}
\end{figure}

\subsubsection{Return Row unit}
\label{sec:return_row_unit}

	To begin a new row of Counter units following the completion of the current row, a Return Row gadget unit must return the frontier of the assembly to the east side of the construction.  After returning the frontier to the east side, it initiates an incrementing Counter unit for the least significant digit.  In cases where there is only one digit column in use, a single special gadget is used instead of the multi-piece unit.  See Figure~\ref{fig:Return}
    
    We define the Return Row unit by creating the following gadgets:
    
\begin{itemize}
\item If $d=1$, create
\begin{align*}
			{\tt Return\_Row\_Single}( 	& \! \left \langle \texttt{return},\texttt{start}\right \rangle \! , \\
									& \! \left \langle \texttt{inc},\texttt{read},1\right \rangle \!, \\
									& \! \left \langle \texttt{d\_fill} \right \rangle \! , \\
									& \! \left \langle \texttt{inc},\texttt{read},0\right \rangle )
\end{align*}
from the general gadget in Figure~\ref{fig:Return_Row_Single}. One gadget was created in this step.
\item Otherwise:

\begin{itemize}

\item Create
	\begin{align*}
		\texttt{Return\_Row\_Start}( 	& \! \left<\texttt{msd\_down\_z\_0},3l - 2\right> \!, \\
								& \! \left \langle \texttt{d\_fill} \right \rangle \!, \\
								& \! \left<\texttt{return},1\right> )
	\end{align*}
from the general gadget in Figure~\ref{fig:Return_Row_Start}. One gadget was created in this step.

\item For each $i=1,\ldots,d-2$, create $\texttt{Return\_Row}\left(\left<\texttt{return},i\right> \! ,\left<\texttt{return},i+1\right>\right)$ from the general gadget in Figure~\ref{fig:Return_Row}. In this step, $d-2 = O(k) = O(\log N)$ gadgets were created.

\item Create 
	\begin{align*}
		\texttt{Return\_Row\_End}( 	& \! \left<\texttt{return},d-1\right> \!, \\
								& \! \left<\texttt{inc},\texttt{read},1\right> \!, \\
								& \! \left<\texttt{inc},\texttt{read},0\right> )
	\end{align*}
from the general gadget in Figure~\ref{fig:Return_Row_End}. One gadget was created in this step.

\end{itemize}

\end{itemize}

Since the number of tiles in each gadget in the Return Row unit is $O(1)$, then based on the above computations, the number of  gadgets (and therefore the number of tile types) created in this subsection is $O\left(  \log N \right)$.

\begin{figure}[htp]
	\begin{center}
	\begin{subfigure}[b]{.49\textwidth}
	\begin{center}
		\includegraphics[width=43.5px]{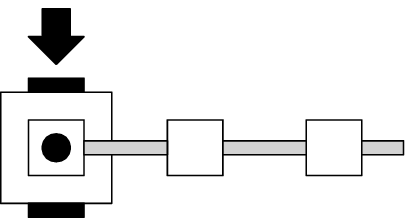}
        \caption{\label{fig:Return_Row_Start} \texttt{Return\_Row\_Start}}
        \vspace{.15\linewidth}
        \includegraphics[width=52.5px]{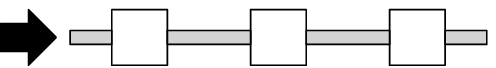}
        \caption{\label{fig:Return_Row} \texttt{Return\_Row} }
        \vspace{.15\linewidth}
        \includegraphics[width=51px]{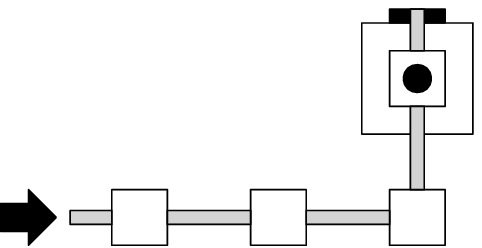}
        \caption{\label{fig:Return_Row_End} \texttt{Return\_Row\_End} }
        \vspace{.15\linewidth}
        \includegraphics[width=42px]{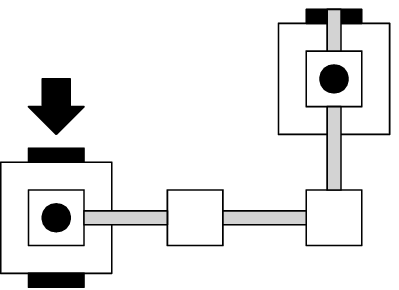}
        \caption{\label{fig:Return_Row_Single} \texttt{Return\_Row\_Single} }
	\end{center}
    \end{subfigure}
    \begin{subfigure}[b]{.49\textwidth}
    \begin{center}
		\includegraphics[width=132px]{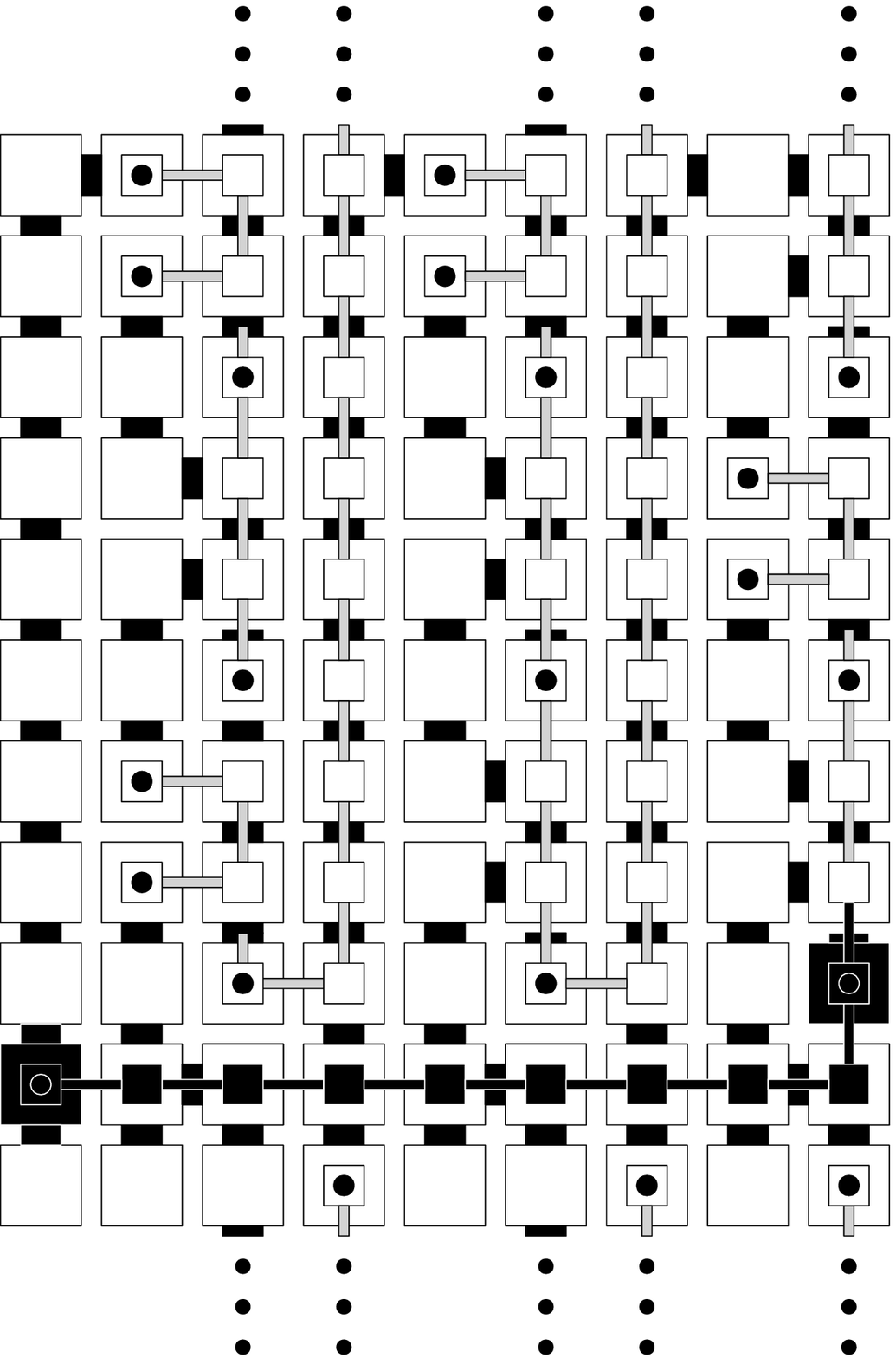}
        \caption{\label{fig:Example_Return_Row} An example with the Return Row unit colored black.}
    \end{center}
    \end{subfigure}
    	\caption{\label{fig:Return}The Return Row gadget unit.}
    \end{center}
\end{figure}

\subsubsection{Roof unit}

\label{sec:roof_unit}

The Roof gadget unit consists of a vertical tile column that reaches above the tiles from the last counter row, then extends the assembly to a height of $N$.  Along the column are glues placed periodically that accept filler extensions.  These extensions are designed to patch up holes that are left in the last counter row.  The highest tile in the column has west-facing and east-facing glues.  These glues accept shingle tiles which extend the roof westward and eastward so that the entire construction is ``covered''.  (The eastward expansion is blocked if $r=0$.)  Each shingle tile has a south-facing glue that binds to a filler tile, which will cover up any remaining gaps in the $k\times N$ shape.  The Roof unit is shown in Figure~\ref{fig:Roof}.

We define the Roof unit by creating the following gadgets:

\begin{itemize}

\item Create $\texttt{Up\_Column}\left(\left<\texttt{inc},\texttt{read},bin(2m-1,l)\right> \! ,\left<\texttt{roof},\texttt{col},2\right>\right)$ from the general gadget in Figure~\ref{fig:Up_Column_Tile}. One gadget was created in this step.

\item Create $\texttt{Up\_Column}\left(\left<\texttt{roof},\texttt{col},2\right> \! ,\left<\texttt{roof},\texttt{col},3\right>\right)$ from the general gadget in Figure~\ref{fig:Up_Column_Tile}. One gadget was created in this step.


\item For each $i=3,\ldots,l+2$, create
	\begin{align*}
		\texttt{Roof\_Chimney}( 	& \! \left<\texttt{roof},\texttt{col},i\right> \!, \\
							& \! \left<\texttt{roof},\texttt{col},i+1\right> \!, \\
							& \! \left<\texttt{roof},\texttt{filler},1\right> )
	\end{align*}
from the general gadget in Figure~\ref{fig:Roof_Chimney}. In this step, $(l+2)-3+1 = O(l) = O(\log m) = O\left( \frac{\log N}{k} \right) = O(\log N)$ gadgets were created.

\item For each $i=1,\ldots,d-1$, create\\$\texttt{Roof\_Filler}\left(\left<\texttt{roof},\texttt{filler},i\right> \! ,\left<\texttt{roof},\texttt{filler},i+1\right>\right)$ from the general gadget in Figure~\ref{fig:Roof_Filler}. In this step, $d - 1 = O(k) = O(\log N)$ gadgets were created. 

\item For each $i=l+3,\ldots,l+r+4$,\\create $\texttt{Up\_Column}\left(\left<\texttt{roof},\texttt{col},i\right> \! ,\left<\texttt{roof},\texttt{col},i+1\right>\right)$ from the general gadget in Figure~\ref{fig:Up_Column_Tile}. In this step, $r+2 = O(l) = O(\log m) = O\left( \frac{\log N}{k} \right) = O(\log N)$ gadgets were created.

\item Create
	\begin{align*}
		\texttt{Roof\_Cap}( 	& \! \left<\texttt{roof},\texttt{col},l+r+5\right> \!, \\
						& \! \left<\texttt{roof},\texttt{r\_shingle},1\right> \!, \\
						& \! \left<\texttt{roof},\texttt{l\_shingle},1\right> )
	\end{align*}
from the general gadget in Figure~\ref{fig:Roof_Cap}. One gadget was created in this step.

\item For each $i=1,\ldots,c+2$, create
	\begin{align*}
			{\tt Roof\_Left\_Shingle}( 	&  \! \left \langle \texttt{roof},\texttt{l\_shingle},i \right \rangle \! , \\
								&  \! \left \langle \texttt{d\_fill} \right \rangle \! , \\
								&  \! \left \langle \texttt{roof},\texttt{l\_shingle},i+1 \right \rangle )
	\end{align*}
	from the general gadget in Figure~\ref{fig:Roof_Left_Shingle}. In this step, $c+2 = O(1)$ gadgets were created.
	
\item If $r>0$, then for each $i=1,\ldots,3d-3$, create
	\begin{align*}
			{\tt Roof\_Right\_Shingle}( 		&\! \left \langle \texttt{roof},\texttt{r\_shingle},i \right \rangle \! , \\
								   		& \! \left \langle \texttt{roof},\texttt{r\_shingle},i+1 \right \rangle \!, \\\
										& \! \left \langle \texttt{d\_fill} \right \rangle )
	\end{align*}
	from the general gadget in Figure~\ref{fig:Roof_Right_Shingle}. In this step, if $r > 0$, then $3d-3 = O(d) = O(k) = O(\log N)$ gadgets were created.

\end{itemize}

Since the number of tiles in each gadget in the Roof is $O(1)$, then based on the above computations, the number of  gadgets (and therefore the number of tile types) created in this subsection is $O(\log N)$.

\begin{figure}[htp]
	\begin{center}
    \begin{subfigure}[b]{.49\textwidth}
    \begin{center}
		\includegraphics[width=13.5px]{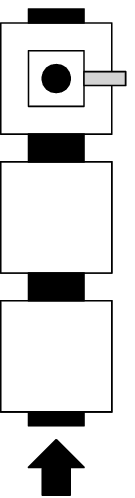}
        \caption{\label{fig:Roof_Chimney}\texttt{Roof\_Chimney}}
        \vspace{.1\linewidth}
        \includegraphics[width=52.5px]{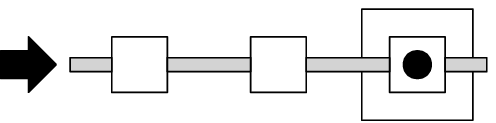}
        \caption{\label{fig:Roof_Filler}\texttt{Roof\_Filler}}
        \vspace{.1\linewidth}
		\includegraphics[width=15px]{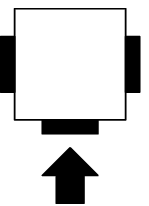}
        \caption{\label{fig:Roof_Cap}\texttt{Roof\_Cap}}
        \vspace{.1\linewidth}
        \includegraphics[width=22.5px]{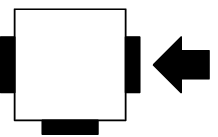}
        \caption{\label{fig:Roof_Left_Shingle}\texttt{Roof\_Left\_Shingle}}
        \vspace{.1\linewidth}
        \includegraphics[width=22.5px]{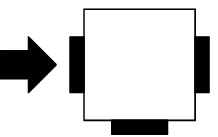}
        \caption{\label{fig:Roof_Right_Shingle}\texttt{Roof\_Right\_Shingle}}
    \end{center}
    \end{subfigure}
    \begin{subfigure}[b]{.49\textwidth}
    \begin{center}
		\includegraphics[width=162px]{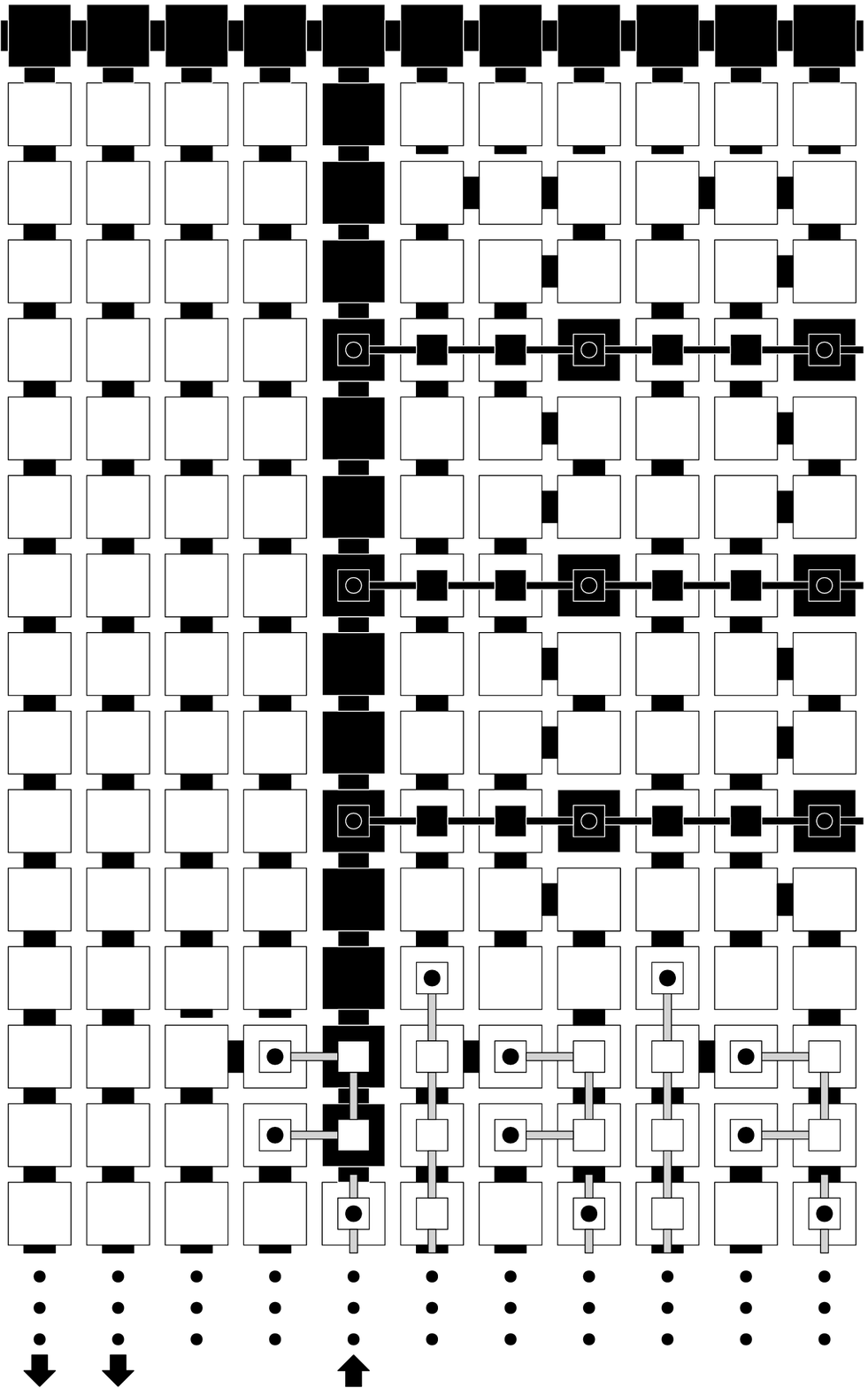}
        \caption{\label{fig:Example_Roof} An example with the Roof unit colored black.}
    \end{center}
    \end{subfigure}
    	\caption{\label{fig:Roof}The Roof gadget unit.}
	\end{center}
\end{figure}

\subsection{Proof of Theorem~\ref{thm:two}}

We use the \emph{conditional determinism} method of Lutz and Shutters \cite{LutzShutters12} to prove that our construction is directed. Conditional determinism is weaker than the local determinism method of Soloveichik and Winfree \cite{SolWin07} but strong enough to imply unique production. Although Lutz and Shutters define conditional determinism for a 2D TAS, their result holds in 3D as well, namely, if $\mathcal{T}$ is a conditionally deterministic TAS, then $\mathcal{T}$ is directed.

We merely give a brief, intuitive review of conditional determinism (see \cite{LutzShutters12} for a thorough discussion). Let $\mathcal{T} = (T,\sigma,\tau)$ be a TAS and $\vec{\alpha}$ be an assembly sequence of $\mathcal{T}$ with result $\alpha$. Conditional determinism relies on the notion of \emph{dependence} (of one position on another position in an assembly). For a position $\vec{m} \in \textmd{dom}\left({\alpha}\right)$ and some unit vector\\$\vec{u} \in \{(0,1,0),(1,0,0),(0,-1,0),(-1,0,0),(0,0,1),(0,0,-1)\}$, we say that $\vec{m}+\vec{u}$ \emph{depends} on $\vec{m}$, or $\vec{m}+\vec{u}$ is a dependent of $\vec{m}$, if, in every assembly sequence in $\mathcal{T}$, a tile is placed at $\vec{m}$ before a tile is placed at $\vec{m}+\vec{u}$. Intuitively, $\vec{\alpha}$ is conditionally deterministic if the following three properties are satisfied: (1) every tile placed by $\vec{\alpha}$ binds with exactly temperature $\tau$, (2) for every position $\vec{m} \in \textmd{dom}\;{\alpha}$, if $\alpha(\vec{m})$ and the union of its immediate output and dependent neighbors (adjacent positions that depend on $\vec{m}$) are removed from $\alpha$ to get $\alpha'$, then $\alpha(\vec{m})$ is the only tile that may bind at position $\vec{m}$ in $\alpha'$, and (3) $\alpha$ is terminal. A TAS is conditionally deterministic if it exhibits a conditionally deterministic assembly sequence. We are now ready to prove Theorem~\ref{thm:two}.

Assume that $\sigma$ is the seed assembly that consists of the single seed tile that was previously specified and let $\mathcal{T}_{k,N} = \left(T_{k,N},\sigma,1\right)$. We are now ready to prove Theorem~\ref{thm:two}. 

\begin{proof}(of Theorem~\ref{thm:two}) \\
First, note that $\left|T_{k,N}\right| = O\left( N^{\frac{1}{\left \lfloor \frac{k}{3} \right \rfloor}} + \log N\right)$ because there are two\\{\tt Vertical\_Column} tile types, a constant number of gadget units (Seed, Counter, Return Row, and Roof), the Counter unit contributes $O\left( N^{\frac{1}{\left \lfloor \frac{k}{3} \right \rfloor}} + \log N \right)$ tile types and the Seed, Return Row and Roof units each contribute $O\left( \log N \right)$ tile types.

Now, let $\vec{\alpha}$ be the assembly sequence in $\mathcal{T}_{k,N}$, with result $\alpha$, such that the order in which $\vec{\alpha}$ places tiles can be inferred from Figure~\ref{fig:Example_Full}, and with the additional constraint that, when building the assembly sequence in a depth-first manner, $\vec{\alpha}$ always tries to place as many tiles as possible in the $z=0$ plane before placing a tile in the $z=1$ plane, breaking ties on which direction in which to proceed based on the ordering: north, east, south and finally west. Thus, $R^3_{k,N}$ self-assembles in $\mathcal{T}_{k,N}$. To show that $\mathcal{T}_{k,N}$ is conditionally deterministic, and therefore directed, it suffices to show that $\vec{\alpha}$ is conditionally deterministic.

To prove that $\vec{\alpha}$ is conditionally deterministic, first note that, by the way we construct all the gadgets, all tiles placed by $\vec{\alpha}$ initially bind deterministically with exactly strength $\tau = 1$. Now let $\vec{m} \in \textmd{dom}\;{\alpha}$ and\\$\vec{u} \in \{(0,1,0),(1,0,0),(0,-1,0),(-1,0,0),(0,0,1),(0,0,-1)\}$ such that $\alpha(\vec{m}+\vec{u})$ is defined. The following statements can be verified by inspection of how the gadgets defined in the previous subsection attach to each other:
\begin{enumerate}
    \item If $\alpha(\vec{m})$ has a strength-$1$ glue in direction $\vec{u}$ and $\alpha(\vec{m}+\vec{u})$ has a strength-$0$ glue in direction $-\vec{u}$, then $\vec{m}$ depends on $\vec{m} + \vec{u}$ in every assembly sequence in $\mathcal{T}_{k,N}$.
    \item If $\alpha(\vec{m})$ has a strength-$1$ glue in direction $\vec{u}$ and $\alpha(\vec{m}+\vec{u})$ has a strength-$1$ glue in direction $-\vec{u}$, then there are two sub-cases to consider, depending on whether the glues interact:
    \begin{enumerate}
    	\item The case where the glues are not equal and therefore do not interact does not happen in our construction. So, ignore this case.
    
    	\item  If the glues are equal, then, by the way we construct the gadgets, in every assembly sequence in $\mathcal{T}_{k,N}$, $\vec{m}+\vec{u}$ is an output side of $\alpha\left(\vec{m}\right)$ or $\vec{m}+\vec{u}$ is the unique input side of $\alpha\left(\vec{m}\right)$.
    \end{enumerate}
\end{enumerate}
We need not consider the cases where $\alpha(\vec{m})$ has a strength-$0$ glue in direction $\vec{u}$ and $\alpha(\vec{m}+\vec{u})$ has a strength-$0$ glue in direction $-\vec{u}$ and either one depends on the other or not. Thus, for every position $\vec{m} \in \textmd{dom}\;{\alpha}$, if $\alpha(\vec{m})$ and the union of its immediate output and dependent neighbors are removed from $\alpha$ to get $\alpha'$, then $\alpha(\vec{m})$ is the only tile that may bind (via its unique input side) at position $\vec{m}$ in $\alpha'$. Finally, since $\alpha$ is terminal, we may conclude that $\vec{\alpha}$ is conditionally deterministic and therefore $\mathcal{T}_{k,N}$ is directed.
\end{proof}

\begin{corollary}
If $R^3_{k,N}$ is a thin rectangle, then $K^1_{USA}\left( R^3_{k,N} \right) = O\left( N^{\frac{1}{\left \lfloor \frac{k}{3} \right \rfloor}} \right)$.
\end{corollary}

\begin{proof}
By the definition of thin rectangle, we have
$$
k < \frac{\log N}{\log \log N - \log \log \log N} < \frac{\log N}{\log \log N - \frac{1}{2}\log \log N} = \frac{2\log N}{\log \log N},
$$
where the second inequality holds for $N > 2^{2^{4}}$.
Then we have
\begin{eqnarray*}
	\log N	& = &  2^{\log \log N} 
			\ = \ \left( N^{\frac{1}{\log N}} \right)^{\log \log N}  
			\ = \ \left( N^{\frac{1}{\frac{\log N}{\log \log N}}} \right)^{\frac{2}{2}} 
			\ = \ N^{\frac{2}{\frac{2 \log N}{\log \log N}}} \\
			& = & O\left( N^{\frac{2}{k}} \right) \ = \ O\left( N^{\frac{1}{\frac{k}{2}}} \right) \ = \ O\left( N^{\frac{1}{\frac{k}{3}}} \right) \ = \ O\left( N^{\frac{1}{\left \lfloor \frac{k}{3} \right \rfloor}} \right). \\
\end{eqnarray*}

\end{proof}

\section{Future work}
\label{sec:future-work}

It is well-known that the tile complexity of a 2D $N \times N$ square at temperature-2 is $O\left( \frac{\log N}{\log \log N} \right)$. 
More formally, for an $N \times N$ square $S^2_N = S_N = \{0,1,\ldots, N-1\}\times\{0,1,\ldots,N-1\}$, $K^2_{USA}\left(S_N\right) = O\left( \frac{\log N}{\log \log N}\right)$ \cite{AdlemanCGH01}. 
However, it is conjectured \cite{ManuchSS10} that $K^1_{USA}\left(S_N\right)=2N-1$, meaning that a 2D $N \times N$ square does not self-assemble efficiently at temperature-1. 
Yet, the tile complexity of a just-barely 3D $N \times N$ square at temperature-1 is $O\left(\frac{\log N}{\log \log N}\right)$. 
That is, $K^1_{USA}\left(S^3_N\right) = O\left(\frac{\log N}{\log \log N} \right)$, where $S^3_N$ is a just-barely 3D $N \times N$ square, satisfying $\{0,1,\ldots,N-1\}\times\{0,1,\ldots,N-1\}\times\{0\} \subseteq S^3_N \subseteq \{0,1,\ldots,N-1\}\times\{0,1,\ldots,N-1\}\times\{0,1\}$ \cite{jFurcyMickaSummers}. 
So, a 2D $N \times N$ square has the same asymptotic tile complexity at temperature-2 as its just-barely 3D counterpart does at temperature-1. 
Regarding thin rectangles, we know that $K^2_{USA}\left(R^2_{k,N}\right) = O\left(N^{\frac{1}{k}}\right)$ \cite{AGKS05g} and we speculate whether a similar upper bound holds for a just-barely 3D $k \times N$ thin rectangle at temperature-1. 
In other words, is it the case that either $K^1_{SA}\left(R^3_{k,N}\right)$ or $K^1_{USA}\left(R^3_{k,N}\right)$ is equal to $O\left(N^{\frac{1}{k}}\right)$? 
If not, then what are tight bounds for $K^1_{SA}\left(R^3_{k,N}\right)$ and  $K^1_{USA}\left(R^3_{k,N}\right)$? 
We conjecture that $K^1_{USA}\left( R^3_{k,N} \right) = o\left( N^{ \frac{1}{\left \lfloor \frac{k}{3} \right \rfloor} }\right)$. 

\section*{Acknowledgement}
We thank Ryan Breuer (University of Wisconsin Oshkosh Computer Science Major, class of 2019) for writing a computer program that implements the construction for our main positive result. Through his implementation, Ryan uncovered a couple technical flaws with an earlier version of the construction, which have since been corrected in version that was presented here. 

\bibliographystyle{amsplain}
\bibliography{tam}

\clearpage

\end{document}